\documentclass[12pt]{amsart}
\usepackage{srcltx}

\newcounter{ourcount}
\advance\textheight\topskip
\usepackage{amsmath,amsfonts,amsthm,amssymb,amscd}
\allowdisplaybreaks

 \usepackage{graphicx}
 \usepackage[curve,matrix,arrow,color]{xy}

 \usepackage[all,dvips,color]{xypic}
 \usepackage[usenames,dvipsnames]{color}
 \xyoption{line}

\usepackage[mathscr]{eucal}

\usepackage[OT2,T1]{fontenc}

\newcommand{\bref}[1]{\textbf{\ref{#1}}}

\newcommand{\W}{W\!\!B}
\newcommand{\half}{%
  \mathchoice{\ffrac{1}{2}}{\frac{1}{2}}{\frac{1}{2}}{\frac{1}{2}}}
\newcommand{\rmi}{\mathrm{i}}

\newcommand{\tensor}{\otimes}

\newcommand{\sigmax}{\sigma^x}
\newcommand{\sigmay}{\sigma^y}
\newcommand{\sigmaz}{\sigma^z}
\newcommand{\sigmap}{\sigma^+}
\newcommand{\sigmam}{\sigma^-}

\newcommand{\modwX}{\mathcal{X}}
\newcommand{\modwP}{\mathcal{P}}
\newcommand{\modwY}{\mathcal{Y}}

\newcommand{\oN}{\mathbb{N}}

\newcommand{\qS}[1]{S_{#1}}

\newcommand{\cas}{\boldsymbol{C}}
\newcommand{\casim}{\widetilde{\cas}}
\newcommand{\idem}{\boldsymbol{e}}

\newcommand{\q}{\mathfrak{q}}
\newcommand{\ffrac}[2]{\mbox{\footnotesize$\displaystyle\frac{#1}{#2}$}}
\newcommand{\dd}{\mathsf{d}}

\newcommand{\modd}{\,\mathrm{mod}\,}

\newcommand{\stf}{t}

\newcommand{\one}{\boldsymbol{1}}

\newcommand{\Wlat}{\mathcal{W}}
\newcommand{\Wlatq}[2]{\Wlat_{#1,#2}}
\newcommand{\Wlatt}[2]{\Wlat'_{#1,#2}}
\newcommand{\UresSL}[1]{\overline{U}_{\q} s\ell(#1)}

\newcommand{\LQG}{U_{\q} s\ell(2)}
\newcommand{\LQGi}{U_{i} s\ell(2)}

\newcommand{\TLq}[1]{\mathcal{TL}_{\q,#1}}
\newcommand{\TLi}[1]{\mathcal{TL}_{i,#1}}

 \newcommand{\gl}{g\ell}
\newcommand{\oC}{\mathbb{C}}
\newcommand{\Endo}{\mathrm{End}}
\newcommand{\Hom}{\mathrm{Hom}}

\newcommand{\Hilb}{\mathcal{H}}
\newcommand{\chVv}{\Hilb_{N}}

\newcommand{\repQGgl}{\rho_{\gl}}
\newcommand{\repQGq}{\rho_{\q,N}}
\newcommand{\repQG}[1]{\rho_{\q,#1}}
\newcommand{\LUresSL}[1]{U_{\q} s\ell(#1)}

\newcommand{\K}{\mathsf{K}}
\newcommand{\F}{\mathsf{F}}

\newcommand{\f}{\mathsf{f}}

\newcommand{\E}{\mathsf{E}}

\newcommand{\h}{\mathsf{h}}
\newcommand{\e}{\mathsf{e}}

\newcommand{\stprp}{\mathsf{a}}

\newcommand{\XX}{\mathsf{X}} 
\newcommand{\repX}{\mathsf{X}}
\newcommand{\PP}{\mathsf{P}}

\newcommand{\IrrTL}[1]{\mathcal{X}_{#1}}

\newcommand{\PrTL}[1]{\mathcal{P}_{#1}}

\newcommand{\Verma}{\mathscr{V}}

\newcommand{\ket}[1]{|#1\rangle}

\newcommand{\oZ}{\mathbb{Z}}

\swapnumbers
\newtheorem{Thm}[subsection]{Theorem}

\newtheorem{lemma}[subsubsection]{Lemma}

\newtheorem{prop}[subsubsection]{Proposition}

\newtheorem{Conj}[subsection]{Conjecture}

\theoremstyle{definition}
\newtheorem{Dfn}[subsection]{Definition}

\begin{document}
\title[W-algebra in XXZ spin-chains]{%
Lattice W-algebras and logarithmic CFTs}

\author{A.M.~Gainutdinov, H.~Saleur and I.Yu.~Tipunin}
\address{AMG: Institut de Physique Th\'eorique, CEA Saclay, Gif Sur Yvette, 91191, France}
\email{azat.gaynutdinov@cea.fr}
\address{HS: Institut de Physique Th\'eorique, CEA Saclay, Gif Sur
  Yvette, 91191, France.\mbox{}\newline
\mbox{}\qquad\quad Department of Physics and Astronomy,
University of Southern California,
Los Angeles, CA 90089, USA}
\email{hubert.saleur@cea.fr}
\address{IYuT: Tamm Theory Division, Lebedev Physics Institute, Leninski pr., 53,
Moscow, Russia, 119991}
\email{tipunin@gmail.com}

\begin{abstract}
This paper is part of an effort to gain further understanding of 2D Logarithmic Conformal Field Theories (LCFTs) by exploring their lattice regularizations. While all work so far has dealt with the Virasoro algebra (or  the product $\hbox{Vir}\otimes \overline{\hbox{Vir}}$), the best known (although maybe not the most relevant physically)  LCFTs in the continuum  are characterized by a W-algebra symmetry, whose presence is  powerful, but  whose role as a ``symmetry'' remains mysterious. We explore here the origin of this symmetry in the underlying lattice models. We consider $\LQG$ XXZ spin chains for $\q$ a root of unity, and
argue that the centralizer of the ``small'' quantum group $\UresSL2$ goes over the W-algebra in the continuum limit. We justify this identification by representation theoretic arguments, and give, in particular, lattice versions of the W-algebra generators. In the case $\q=i$, which corresponds to symplectic fermions at central charge $c=-2$, we provide a full analysis of the scaling limit of the lattice Virasoro and W generators, and show in details how the corresponding continuum Virasoro and W-algebras are obtained. Striking similarities between the lattice W algebra and the Onsager algebra are observed in this case.

\end{abstract}

\maketitle

\section{Introduction}

It is not surprising that some crucial algebraic objects should play a role both in lattice models and in conformal field theories. For instance, the  Yang--Baxter equation expresses  factorizability of lattice interactions and determines integrable Boltzmann weights; it also expresses factorizability of braiding in the CFT, and is related with monodromy properties of conformal blocks.

Work of the last many years however has unraveled deeper and less expected connections.
Among these was the surprising observation that  order parameters (local height probabilities) are given by CFT branching functions  \cite{Date},
or that the representation theory of lattice algebras such as the Temperley--Lieb algebra is equivalent,
in a certain categorical sense (see \cite{ReadSaleur07-1,ReadSaleur07-2,GRS1,GV,GJSV} for more details),
to the representation theory of the Virasoro algebra (see also \cite{PRZ,RP}).
This has proven particularly useful in understanding logarithmic conformal field theories (LCFTs), a topic which has gained a lot of attention recently.

Indeed, classifying and solving LCFTs from first principles seems very difficult. This is for two main reasons. On the one hand, the loss of unitarity allows a proliferation of  new universality classes. On the other hand, many of the tools that were available in the unitary case -- such as the correspondence between bulk and boundary theories, or the full factorization of the operator product expansions -- are not relevant any longer, or need deep modifications. It seems therefore most reasonable therefore to try and gain additional understanding of the problem by resorting to other approaches than pure bootstrap,  conformal field theory ones. The consideration of lattice regularizations is most useful in that respect, and allows one, in particular,  to rely on many advanced results of algebra. The representation theory of non semi-simple cellular algebras for instance has allowed  in this way great progress in the  understanding of the types of
indecomposable modules or fusion rules appearing in many classes of LCFTs. Further general discussion of this `lattice approach' can be found in \cite{Review}.

It is crucial to try to understand the observations in~\cite{PRZ,ReadSaleur07-2} more thoroughly, building on the intuitive notion that the algebra of local Hamiltonians (e.g., nearest coupling for spin chains) should in some sense go over to the algebra of local Hamiltonians (the stress energy tensor) in the continuum limit, and for boundary theories, where only one chiral algebra is expected. However, showing precisely  how, say,  the Virasoro algebra emerges from the Temperley--Lieb algebra in the continuum limit in general remains a very hard exercise. This is fully understood so far in the   Ising model case only, where the presence of an underlying Majorana fermion makes things quite easy~\cite{KooSaleur}.

It is also worth mentioning that among LCFTs, the rational class  with ``large'' chiral algebras like the triplet W-algebras~\cite{GK,[FGST3]} has a finite
number of primary fields (or isomorphism classes of simple modules) and is particularly accessible for rigorous algebraic
analysis~\cite{[HLZ],[AM1]}, as well as for an explicit treatment of their  Verlinde algebra aspects~\cite{[FHST],[GabR],[GT],[az1],[RW]} and hopefully
mapping class group properties~\cite{FSS}.
 The simplest such LCFT occurs in the so called symplectic fermion theory at $c=-2$~\cite{Kausch,GK} where the chiral
algebra is generated by a triplet of fermion bilinears of conformal weight $h=3$:
\begin{eqnarray*}
W^+&=&\partial\eta^+\eta^+,\nonumber\\
W^0&=&\ffrac{1}{2}\left(\partial\eta^+\eta^-+\partial\eta^-\eta^+\right),\nonumber\\
W^-&=&\partial\eta^-\eta^-
\end{eqnarray*}
satisfying the non linear relations
\begin{eqnarray}\label{w-ope}
W^\alpha(z)W^\beta(w)&=&g^{\alpha\beta}\left({\frac{1}{2}\over (z-w)^6}-\frac{3}{2}{T(w)\over (z-w)^4}
-{3\over 4} {\partial T(w)\over (z-w)^3}\right.\nonumber\\
&&\left. +{3\over 8} {\partial^2 T(w)\over (z-w)^2}-2{(T^2)(w)\over (z-w)^2}+{1\over 4}
{\partial^3 T(w)\over z-w}-{\partial (T^2)(w)\over z-w}\right)\\
&&+5 f^{\alpha\beta}_\gamma\left({W^\gamma(w)\over (z-w)^3}+{1\over 2}{\partial W^\gamma(w)\over (z-w)^2}+{1\over 25}
{\partial^2 W^\gamma(w)\over z-w}+{12\over 25}{(TW^\gamma)(w)\over z-w}\right)\nonumber
\end{eqnarray}
with $g^{+-}=g^{-+}=4$, $g^{00}=2$, $f^{+-}_0=-f^{-+}_0=2$, $f^{0+}_+=-f^{+0}_+=-f^{0-}_-=f^{-0}_-=1$ and other $g^{\alpha\beta}$
and $f^{\alpha\beta}_\gamma$  equal to $0$.

Most studied generalizations concern the so called $(1,p)$ models (symplectic fermions being the $(1,2)$ model)  where the chiral algebra is the triplet W-algebra $\Wlat(p)$ defined first in~\cite{[K-first]} as an extension of the Virasoro algebra by a triplet of primary fields of conformal dimension $2p-1$.
Identifying lattice models which  have such W-algebra $\Wlat(p)$ symmetries remains, in our opinion and except for $p=2$,
a challenge, even though some formal progress has been made in this direction in the last
few years~\cite{Rasmussen08}.

The study of the $W$-algebras is notoriously difficult,  not only because of the high non linearity of the OPEs,
but also because  they  do lack physical understanding. Motivated by the  crucial role played by
these algebras in the rational class of LCFTs, we are presenting here a fresh look at this question by identifying the lattice
equivalents of the tripet $W$-algebras $\Wlat(p)$ in $\LQG$ symmetric XXZ spin chains when $\q$ is the $2p$th primitive root of unity.
The idea of their definition is based on the full symmetry algebra (or {\sl centralizer}) paradigm ``The quantum groups relevant to
quantum-integrable models, \textit{e.g.,} CFTs and their lattice discretizations are one and the same.''
 So, the lattice equivalent of the Virasoro algebra could be defined using its centralizer in CFT~\cite{ReadSaleur07-2,[BFGT],[BGT]} -- the {\sl full} quantum group $\LQG$ -- as the algebra of local Hamiltonian densities commuting with the $\LQG$.
Indeed, such an algebra on the finite spin-chain  is the Temperley--Lieb (TL) algebra and it has close relations with the representation
theory of the Virasoro algebra~\cite{ReadSaleur07-1,GV}.
Moreover, we will show in this paper that particular elements of the TL algebra in the XX spin-chain converge
in the scaling limit to the Virasoro modes.
To go further and define  lattice algebras that would converge to the chiral triplet W-algebras,
we use the (so called Kazhdan--Lusztig) duality~\cite{[FGST],[FGST2]} between these chiral algebras  and the {\sl small} or ``restricted'' quantum group denoted by $\UresSL2$.
This duality essentially states\footnote{The equivalence between the categories of the triplet W-algebra modules and
finite-dimensional $\UresSL2$-modules was proven first in~\cite{[FGST2]} for the $(1,2)$ models and in~\cite{Tsuchiya} for all $(1,p)$ cases.
The functor realizing such an equivalence could be explicitly constructed using a bimodule over the mutual centralizers, the triplet $W$-algebra and the
small quantum group.} that the centralizer of the $W$-algebra $\Wlat(p)$ in a particular Hilbert space of CFT states is the $\UresSL2$.
Therefore in this case, one could think about the lattice equivalent as the centralizer of the small quantum group on a finite spin-chain.

The paper is organized as follows. Sections 2 and 3 are devoted to the case $\q=i$, that is the XX
spin chain (with boundary terms), which is equivalent to a $\gl(1|1)$ spin chain built on alternating
representations. Using  the earlier proposal in~\cite{KooSaleur,GRS1},  a lattice version of the Virasoro
algebra is obtained  by considering TL generators  -- recall that the TL algebra
is the centralizer of the full quantum group acting on the spin chain.   In section 2, we  propose
 similarly a lattice version of the $W$-algebra obtained as the centralizer of the small
quantum group $\UresSL2$ acting on the spin chain. Fourier modes (particular linear combinations of the lattice $W$-algebra generators) which constitute ``approximations''
of the Virasoro or chiral $W$-algebra modes are then defined, and their commutators are studied in  section 2 as well.
Of course, for a finite system, these approximations   do not form a closed  algebra.
Calculating further commutators leads to infinitely many such ``approximations'', and a
complicated algebraic structure, which is studied in detail. The scaling limit is then
studied in section 3 where we are able, after the introduction of appropriate normal ordering,
to obtain exactly, from the lattice and in the limit of large chains, the defining relations of
the Virasoro and $W$ algebras for symplectic fermions. In section 4, we generalize our analysis
to other values of $\q$. We define the lattice $W$ algebra for the XXZ spin chain,
and study its representation theory in details. While refraining from a full analysis
of the continuum limit, we show that our results are fully consistent with what is
expected in the $(p-1,p)$ and $(1,p)$ models of  LCFTs. Section 5 contains some conclusions and pointers for future work. Five appendices contain supplementary materials like long definitions, proofs and bulky expressions.


\section{The triplet W-algebra from XX spin-chain}
In this section, we give an explicit definition of the lattice
analogue of Kausch's triplet W-algebra in the case of XX
spin-chains. The use of the well-known fermionic formulation of the model
allows to introduce all modes of ``W-currents'' on the lattice. Their
scaling limit and commutation relations are then studied in the next
section.

\subsection{The XX model}
We consider the XX model Hamiltonian of $N$ one-half spins with an open
boundary condition described by the ``quantum-group symmetric'' boundary term~\cite{PasquierSaleur},
\begin{equation}\label{XX-sigma}
H = \half\sum_{j=1}^{N-1}(\sigmax_j \sigmax_{j+1} + \sigmay_j \sigmay_{j+1}) - \ffrac{i}{2}(\sigmaz_1-\sigmaz_{N}).
\end{equation}
This is an operator acting on $\chVv=\oC^{2\otimes N}$
and $\sigmax_i, \sigmay_i$ and $\sigmaz_i$ are usual Pauli matrices,
\begin{equation}\label{Pauli}
\sigmax =
\begin{pmatrix}
0 & 1\\
1 & 0
\end{pmatrix}, \quad
\sigmay =
\begin{pmatrix}
0 & -i\\
i & 0
\end{pmatrix}, \quad
\sigmaz =
\begin{pmatrix}
1 & 0\\
0 & -1
\end{pmatrix}.
\end{equation}

\subsection{Lattice fermions}
It is
 useful in what follows to reformulate everything in terms of ordinary
 lattice fermions $c_j$ and $c^{\dagger}_j$.
We use the
 Jordan-Wigner transformation to get the lattice fermions %
 \begin{equation}\label{JW-trans}
\begin{split}
 c_j^\dagger&=i^{j-1}~i^{\sigma_1^z+\ldots+\sigma_{j-1}^z}\otimes{\sigma_j^+},\\
  c_j&=
 i^{-j+1}~i^{-\sigma_1^z-\ldots-\sigma_{j-1}^z}\otimes {\sigma_j^-}
\end{split}
 \end{equation}
with the anticommutation relations
\begin{equation}
 \{c^\dagger_j,c_{j'}\}=\delta_{j,j'}.
\end{equation}

This transformation gives the free fermion expression for the XX
Hamiltonian from~\eqref{XX-sigma}:
\begin{equation}\label{XXH}
 H=-\sum_{j=1}^{N-1}e_j = \sum_{j=1}^{N-1} \bigl(c^\dagger_jc_{j+1}+c^\dagger_{j+1}c_j\bigr)
-i c^\dagger_1 c_1+i c^\dagger_N c_N,
\end{equation}
where we introduce the Hamiltonian densities
\begin{equation}\label{eq:PTL-rep-first}
e_j=c_j c_{j+1}^\dagger + c_{j+1} c_j^\dagger
  +i\bigl(c_j^\dagger c_j-c^\dagger_{j+1}c_{j+1}\bigr),
\qquad 1\leq j\leq N-1,
\end{equation}
which satisfy  defining relations for the Temperley--Lieb algebra with
zero fugacity parameter:
\begin{align*}
e_i^2 &= 0,\\
e_i e_{i\pm1}e_i &= e_i,\\
e_i e_j &= e_j e_i,\quad |i-j|>1.
\end{align*}
We denote the algebra generated by these $e_j$, with $1\leq j\leq N-1$,
and the identity as $\TLi{N}$.

\subsection{Quantum group results for $\q=i$ case}
 For applications to open XX
spin-chains, we  set  in what follows
$\q\equiv i$. As a module over the quantum group $\LQG$, see general definitions in App.~\bref{app:qunatm-gr-def}, the spin
chain $\chVv$ is a tensor product of
two-dimensional irreducible representations such that $\E\to\sigma^+$,
$\F\to\sigma^-$, $\K\to \q\sigma^z$, and $\e=\f=0$, see precise
definitions of quantum groups at roots of unity in App.~\bref{app:qunatm-gr-def}.
Using the $(N-1)$-folded
comultiplications~\eqref{N-fold-comult-cap}
together
with the Jordan-Wigner transformation~\eqref{JW-trans}, we obtain
the representation $\repQGgl:\LQG\to\Endo_{\oC}(\chVv)$ (usual fermionic expressions)
\begin{equation}
\repQGgl(\E)\equiv\Delta^{N-1}(\E) = \sum_{j=1}^{N}\q^j c_j^{\dagger} \repQGgl(\K),
\qquad \repQGgl(\F)\equiv\Delta^{N-1}(\F) = \sum_{j=1}^{N}\q^{j-1} c_j,\label{QG-ferm-1}
\end{equation}
and
\begin{equation}\label{sl2-gen}
\begin{split}
\repQGgl(\e)\equiv\Delta^{N-1}(\e) &= \sum_{1\leq j_1<j_2\leq N}(-1)^{j_1+j_2}\q^{1-j_1-j_2} c_{j_1}^{\dagger}c_{j_2}^{\dagger},\\
\repQGgl(\f)\equiv\Delta^{N-1}(\f) &= \sum_{1\leq j_1<j_2\leq
  N}\q^{j_1+j_2-1} c_{j_1}c_{j_2}.
\end{split}
\end{equation}

We  can then easily check that all the generators $e_j$ of $\TLi{N}$
commute with the quantum-group
action~\cite{PasquierSaleur}. Moreover, the quantum-group generators
give the full symmetry of the XX model, or in more technical terms
the full quantum group $\LQG$ at $\q=i$ gives \textit{the centralizer} of $\TLi{N}$. This is a particular case of a more general
(and, of course, well-known) statement discussed in Sec.~\bref{sec:XXZ} in the context of XXZ representations.

\subsection{Lattice W-algebra: generators and their relations}\label{sec:lat-W-XX}
As was mentioned in the introduction, our paradigm in definition of
lattice algebras (as well as chiral algebras in the scaling limit) is
based on their centralizers, which are various quantum groups. We
will see in Sec.~\bref{sec:scal-lim} that a properly defined scaling
limit of $\TLi{N}$ gives all the Virasoro modes $L_n$ in the symplectic
fermions representation, a logarithmic CFT with the central charge
$c=-2$. It was discussed in~\cite{ReadSaleur07-2,[BFGT]} that the generators of the full quantum group $\LQGi$
give also the centralizer of the Virasoro representation, or more formally,
an equivalent ``staircase'' bimodule structure for the commuting actions of
the Virasoro at $c=-2$ and $\LQGi$ is obtained as a  semi-infinite version of
the finite ones for $\TLi{N}$
and $\LQGi$ extracted from the open XX spin-chains.

In the case of the chiral
triplet W-algebra, it was proven in~\cite{[FGST2]} that the
centralizer of the W-algebra in the symplectic fermions theory  is given by
the small (or restricted) quantum group $\UresSL2$ generated by the
capital generators $\E$ and $\F$, with the Cartan element $\K$, only. By
an analogy with the Temperley--Lieb case, we then define the lattice analog
of  Kausch's triplet W-algebra as an extension of the
Temperley--Lieb algebra for $\q=i$ by all operators commuting with the
action of the restricted quantum group. The new generators will break the $s\ell(2)$
symmetry generated by the divided powers $\e$ and $\f$.  Put more formally, we define
\textit{the lattice W-algebra} $\Wlatq{i}{N}$ as the centralizer (algebra of
all operators intertwining  the action) of the restricted quantum group
$\UresSL2$ on the tensor-product representation~$\chVv$.

In order to describe the additional generators extending $\TLi{N}$ to
$\Wlatq{i}{N}$, we first recall the definition of the walled Brauer
algebra~\cite{Martin-walled} in the context of alternating $\gl(1|1)$ spin-chains introduced
in~\cite{ReadSaleur01}.

\subsubsection{A relation with $\gl(1|1)$ spin-chains and the walled Brauer algebra}

Up to some simple transformations, the XX spin chain  can be reformulated as a $\gl(1|1)$ spin
chain with alternating fundamental representation and its dual.
This is discussed in details for instance in the introductory part
of~\cite{GRS1} where it is shown, among other things, that the
$\gl(1|1)$ symmetry is generated by the
quantum-group generators $\E$, $\F$, and the $s\ell(2)$ Cartan element
$\h$ via
\begin{align}
&\repQGgl(\E)\equiv\Delta^{N-1}(\E) =  F^\dagger_{(1)}\,\repQGgl(\K),\notag\\
 &\repQGgl(\F)\equiv\Delta^{N-1}(\F) = \q^{-1}F_{(1)},\label{QG-fermf-1}\\
&\repQGgl(\K)\equiv\Delta^{N-1}(\K) = (-1)^{\repQGgl(2\h)},\notag
\end{align}
where $F^\dagger_{(1)},F_{(1)}$ are the fermionic generators  and $2\h$ the
fermion number (the last generator of $\gl(1|1)$ is trivially zero in the alternating product of fundamental and its  dual).

We note now that the $\gl(1|1)$ action commutes
with more than the nearest-neighbor coupling  -- Temperley Lieb generators.
It commutes in particular with permutations $\W_j\equiv P_j P_{j+1} P_j$ of next-nearest neighbor $\oC^{1|1}$ tensorands ($P_j$ is the usual flip $a\otimes b\mapsto b\otimes a$ of states on $j$th and $(j+1)$th sites)
which in mathematical terms corresponds to an action of the walled Brauer algebra~\cite{Martin-walled}
(see also Sec.~5 of~\cite{Canduetal}) generated by
$e_j$ and $\W_j$.
The walled Brauer generators are represented in our
spin-chain by
\begin{equation}\label{wall-gen}
 \W_j=(c_j+c_{j+2})(c^\dagger_j+c^\dagger_{j+2})-1,\qquad j=1,2,\dots,N-2.
\end{equation}

It was proven in~\cite{Serg} that this representation of the walled
Brauer algebra gives the centralizer of the $\gl(1|1)$ action in the alternating  tensor-space representation. We
thus see that $\W_j\in\Wlatq{i}{N}$ because of the relations~\eqref{QG-fermf-1} which show the the $\gl(1|1)$ representation contains the image $\repQGgl\bigl(\UresSL2\bigr)$ of the restricted quantum group~$\UresSL2$. Obviously, $\W_j$'s do not
commute with the $s\ell(2)$ generators $\e$ and $\f$ but still commute
with the commutator $[\e,\f]=2\h=S^z$. In order to find the full
centralizer of the quantum group $\UresSL2$ with the Cartan element $\K$ represented by the
operator $(-1)^{2\h}$, we then introduce two more
generators by
\begin{equation}\label{sl2act-WB}
 \W^+_j = [\e,\W_j],\qquad  \W^-_j = [\f,\W_j].
\end{equation}
These operators now break the $S^z$-symmetry (the $\gl(1|1)$ generator $\h$ which is not in $\UresSL2$) but respect the action of
the Cartan generator $\K$ of $\UresSL2$ because $\e$ and $\f$ commute
with~$\K$, see App.~\bref{app:qunatm-gr-def}.

\subsubsection{Definition of the lattice W-algebra}
We conclude that the lattice W-algebra $\Wlatq{i}{N}$ can be defined as an extension of the $\TLi{N}$ algebra by the generators
\begin{gather}\label{wbp}
 \W^+_j=(-1)^{j}\bigl(c^\dagger_jc^\dagger_{j+1}+i c^\dagger_jc^\dagger_{j+2}
-c^\dagger_{j+1}c^\dagger_{j+2}\bigr),\\
\W^0_j=-\ffrac{1}{2}(1 + i c^\dagger_j c_{j+1} - c^\dagger_j c_{j+2} + i c^\dagger_{j+1} c_j -
   2 c^\dagger_{j+1} c_{j+1} - i c^\dagger_{j+1} c_{j+2} - c^\dagger_{j+2} c_j -
  i c^\dagger_{j+2}c_{j+1}),\\
\W^-_j=(-1)^{j+1}\bigl(c_jc_{j+1}+i c_jc_{j+2}
-c_{j+1}c_{j+2}\bigr),\label{wbm}
\end{gather}
where $j=1,2,\dots,N-2$. Of course, this statement should be proven. We formulate our result in Thm.~\bref{thm:centralizer} and give a computational proof in App.~\bref{app:Ures-centr} where we use more convenient fermions (the Fourier transforms of $c_j$'s and $c_j^{\dagger}$'s) introduced below.

We note that the relation of $\W^0_j$ generators with the walled Brauer generators~\eqref{wall-gen} is
\begin{equation}
 \W^0_j=-\ffrac{1}{2}\W_j+\ffrac{i}{2}e_j-\ffrac{i}{2}e_{j+1}.
\end{equation}

\subsubsection{$s\ell(2)$ structure} The introduction of  the $\W^0_j$ generators  is  convenient for  having a triplet of generators with
respect to the action of the $s\ell(2)$ part of $\LQG$ given in~\eqref{sl2-gen}.
 Indeed, we note
 that the generators $\W^+_j$, $\W^0_j$ and $\W^-_j$ for each site~$j$ form an $s\ell(2)$ triplet:
\begin{align}
 [\e,\W^+_j]&=0,&&[\e,\W^0_j]=-\W^+_j,&&[\e,\W^-_j]=2\W^0_j,\label{W-sl2-1}\\
 [\f,\W^+_j]&=-2\W^0_j,&&[\f,\W^0_j]=\W^-_j,&&[\f,\W^-_j]=0,\label{W-sl2-2}
\end{align}
and themselves satisfy $s\ell(2)$ relations on each site
\begin{equation*}
 [\W^+_j,\W^-_j]=2\W^0_j,\qquad [\W^0_j,\W^\pm_j]=\pm\W^\pm_j.
\end{equation*}
We note that these statements agree with a more general statement given below
in~\bref{prop:sl2-Wlatt} for all roots of unity $\q=e^{i\pi/p}$.

The formulas~\eqref{W-sl2-1} and~\eqref{W-sl2-2} mean that we have an action of $s\ell(2)$ on the algebra
$\Wlatq{i}{N}$ by derivatives of the multiplication. Because the generators belong to spin-$0$ and spin-$1$ representations of the $s\ell(2)$, the algebra $\Wlatq{i}{N}$  is therefore decomposed onto $s\ell(2)$-modules of integer spins only. The multiplicity
$m_k$ of a $k$-dimensional irreducible $s\ell(2)$ submodule in $\Wlatq{i}{N}$ is
\begin{equation*}
 m_k
  =\binom{2N-2}{N-k}-\binom{2N-2}{N-k-2},
\end{equation*}
where $k$ is odd.
We note that $m_1$ is the Catalan number which equals the dimension of
the TL subalgebra. This agrees with the fact that the span of $\LQG$-invariants is the  subalgebra $\TLi{N}$.  We finally obtain the dimension of the
$\Wlatq{i}{N}$ algebra
\begin{equation}
\dim(\Wlatq{i}{N}) = \sum_{\mbox{\tiny$k$--odd}}km_k = 2^{2N-3}.
\end{equation}
The numbers $m_k$ and the dimension of the algebra can be computed explicitly using the representation theory of $\Wlatq{i}{N}$ which will be discussed briefly below. We just note now that the number $2^{2N-3}$ obtained as the dimension of the algebra follows from the fact that the dimension of its centralizer -- $\repQGgl(\UresSL2)$ -- is $8=2^3$ while the full matrix algebra (or the Clifford algebra) has dimension $2^{2N}$.

\subsubsection{Relations in the algebra}
It would be interesting of course  to express our algebra $\Wlatq{i}{N}$ directly in terms
of generators and defining  relations, much as can be done for the Temperley--Lieb algebra,
but we have not managed to do so. Investigation of what the relations involving $\W^{\pm,0}_j$
generators might be is quite intricate, and can be organized in terms of the number of sites involved.
We give examples of the relations in App.~\bref{App:W-rel}. We hope that there exist a bigger algebra with simpler defining relations (like for TL or Hecke algebras) such that its spin-chain representation is not faithful and this algebra factors through $\Wlatq{i}{N}$.

\subsection{Fermionic modes}\label{sec:ferm-modes}

\subsubsection{Hamiltonian spectrum and Jordan blocks}
In order to study the behavior of the lattice W-algebra in the continuum limit, we have to first discuss
technical preliminaries. Fourier transform will play an essential role in what follows, and
the corresponding definitions are better understood by first studying  here the spectrum,
eigenvectors and Jordan blocks structure of the Hamiltonian in~\eqref{XXH}.
We first consider the case of even $N$ and
compute eigenvectors of the adjoint action of $H$ in the
one-particle sector with the result
\begin{equation}\label{H-spec}
 [H,\theta^\dagger_k]=2\cos\pi\ffrac{k}{N}\theta^\dagger_k,\qquad
 [H,\theta_k]=-2\cos\pi\ffrac{k}{N}\theta_k,
\end{equation}
where we introduce fermions in the momentum space
\begin{equation}\label{theta-ferm}
 \theta^\dagger_k=\sum_{j=1}^N A_k(j)c^\dagger_j,\qquad
\theta_k=-i\sum_{j=1}^N A_k(j)c_j,\qquad k=1,\dots, N-1
\end{equation}
with the amplitudes
\begin{equation}\label{amplitudes}
 A_k(j)=\ffrac{1}{2}(1+i e^{-i\pi\frac{k}{N}})e^{i\pi\frac{kj}{N}}-
\ffrac{1}{2}(1+i e^{i\pi\frac{k}{N}})e^{-i\pi\frac{kj}{N}},\qquad k=1,\dots, N-1.
\end{equation}
These operators satisfy the anti-commutation relations
\begin{equation}
 \{\theta^\dagger_k,\theta_{k'}\}= N\cos\pi\ffrac{k}{N}\delta_{k,k'}.
\end{equation}

There are also two root vectors
\begin{align}
 \gamma^+=i\ffrac{\sqrt{\pi}}{N}\sum_{j=1}^{N}\Bigl(\ffrac{N}{2}-j
 +\ffrac{1}{2}\bigl(1-(-1)^j\bigr)\Bigr)i^j c^\dagger_j,\label{gammap}\\
\gamma^-=-\ffrac{\sqrt{\pi}}{N}\sum_{j=1}^{N}\Bigl(\ffrac{N}{2}-j+\ffrac{1}{2}\bigl(1-(-1)^j\bigr)\Bigr)i^j c_j\label{gammam}
\end{align}
which transform under the action of $H$ as
\begin{equation*}
 [H,\gamma^+]=-\ffrac{2\sqrt{\pi}}{N}\theta^\dagger_{N/2},\qquad [H,\gamma^-]=-\ffrac{2\sqrt{\pi}}{N}\theta_{N/2}
\end{equation*}
and satisfy the relations
\begin{equation}\label{gamma-rel}
 \{\theta^\dagger_{N/2},\gamma^-\}=\sqrt{\pi},\qquad\{\theta_{N/2},\gamma^+\}=-\sqrt{\pi}.
\end{equation}

Using the spectrum~\eqref{H-spec} of $H$, we conclude that the vacuum or ground state is
\begin{equation}\label{vacuum}
\ket{0}=\theta^\dagger_{N/2}\theta^\dagger_{N/2+1}
 \dots\theta^\dagger_{N-1}\ket{\!\downarrow\dots\downarrow},
\end{equation}
where $\ket{\!\downarrow\dots\downarrow}$ is a reference state with all spins down, and the ``log-partner'' of the vacuum
\begin{equation}
\widetilde{\ket{0}}=\gamma^+\theta^\dagger_{N/2+1}
 \dots\theta^\dagger_{N-1}\ket{\!\downarrow\dots\downarrow},
\end{equation}
i.e., we have $H \ket{0} =E_0\ket{0}$ and $H \widetilde{\ket{0}} =E_0\widetilde{\ket{0}} + \frac{2\sqrt{\pi}}{N} \ket{0}$,
where the ground-state energy $E_0$ is computed below. We thus have a Jordan cell of rank $2$.
It is easy to see using the fermions that all excited states involved in the Hamiltonian's Jordan cells are of maximum rank $2$.

\subsubsection{The spectrum generating algebra}
For later convenience in the scaling limit section, we introduce new fermions
\begin{equation}\label{etapm}
\begin{array}{l}
 \eta^+_n=\sqrt{\frac{n}{N \sin\pi\frac{n}{N}}}\theta^\dagger_{N/2+n},\\
 \eta^-_n=\sqrt{\frac{n}{N \sin\pi\frac{n}{N}}}\theta_{N/2-n},
\end{array}
\qquad n=-\ffrac{N}{2}+1,\dots,\ffrac{N}{2}-1,\quad n\neq0,
\end{equation}
together with the fermionic zero modes
\begin{equation}\label{etapm-zero}
 \eta^+_0=\ffrac{1}{\sqrt{\pi}}\theta^\dagger_{N/2},
 \qquad \eta^-_0=\ffrac{1}{\sqrt{\pi}}\theta_{N/2}.
\end{equation}
We have for these generators the anti-commutation relations
\begin{equation}\label{eta-comm}
 \{\eta^+_n,\eta^-_{n'}\}=n\delta_{n+n',0},\qquad
 \{\eta^\pm_0,\gamma^\mp\}=\pm1.
\end{equation}
We note that the operators $\eta^\pm_n$ become in the scaling limit studied in Sec.~\bref{sec:scal-lim}
the modes of a pair of symplectic fermions fields.

The Hamiltonian~\eqref{XXH} takes then the normal ordered form
\begin{equation}\label{XXH-eta}
 H=
\sum_{n=1}^{N/2-1}\ffrac{2}{n}\sin\pi\ffrac{n}{N}\bigl(\eta^+_{-n}\eta^-_n-\eta^-_{-n}\eta^+_{n}\bigr)
+\ffrac{2\pi}{N}\eta^+_0\eta^-_0+\cot\ffrac{\pi}{2N}-1,
\end{equation}
with the ground state energy
\begin{equation}\label{vacuum-energy}
 H\ket{0}=(\cot\ffrac{\pi}{2N}-1)\ket{0},
\end{equation}
where the vacuum now satisfies
\begin{equation}
 \eta^\pm_{\geq0}\ket{0}=0.
\end{equation}
 We call the algebra generated by these fermions the spectrum generating algebra.
Note that the adjoint action by the Hamiltonian is
\begin{equation}
 [H,\eta^\pm_n]=\pm2\sin\pi\ffrac{n}{N}\eta^\pm_n,\quad[H,\gamma^\pm]=-\ffrac{2\pi}{N}\eta^\pm_0.
\end{equation}


\subsubsection{Quantum-group generators in terms of Fourier modes}
The quantum group generators introduced in~\eqref{QG-ferm-1}
and~\eqref{sl2-gen} are expressed in terms of the
$\eta^{\pm}$-fermions as
\begin{equation}
 \E=\eta^+_0\K,\qquad \F=\eta^-_0
\end{equation}
with the $s\ell(2)$ generators
\begin{gather}
 \e=\sum_{n=1}^{N/2-1}\ffrac{1}{n}\eta^+_{-n}\eta^+_n+\gamma^+\eta^+_0,\qquad
 \f=-\sum_{n=1}^{N/2-1}\ffrac{1}{n}\eta^-_{-n}\eta^-_n-\gamma^-\eta^-_0,\label{sl-def-1}\\
\h=-\ffrac{1}{2}\Bigl(\sum_{n=1}^{N/2-1}\ffrac{1}{n}
\bigl(\eta^+_{-n}\eta^-_n+\eta^-_{-n}\eta^+_n\bigr)
+\gamma^-\eta^+_0+\gamma^+\eta^-_0\Bigr).\label{sl-def-2}
\end{gather}

We then have  the $s\ell(2)$ action on the fermions
\begin{equation}\label{sl-act-1}
 [\e,\eta^-_n]=\eta^+_n,\quad[\f,\eta^+_n]=\eta^-_n,\quad[\h,\eta^\pm_n]=\pm\ffrac{1}{2}\eta^\pm_n,
\end{equation}
\begin{equation}\label{sl-act-2}
 [\e,\gamma^-]=\gamma^+,\quad[\f,\gamma^+]=\gamma^-,\quad[\h,\gamma^\pm]=\pm\ffrac{1}{2}\gamma^\pm.
\end{equation}

\subsubsection{TL and W generators in terms of Fourier modes} Using the transformation~\eqref{theta-ferm},
the TL generators~\eqref{eq:PTL-rep-first} in terms of fermions~\eqref{etapm}  take the form
\begin{equation}\label{eta-TL}
 e_{k}= \ffrac{4(-1)^{k+1}}{N}\!\!
\sum_{j_1,j_2=-\frac{N}{2}+1}^{\frac{N}{2}-1}
     \sqrt{\ffrac{\sin\pi j_1/N}{j_1}}\sqrt{\ffrac{\sin\pi j_2/N}{j_2}}\sin\pi k\Bigl(\half+\ffrac{j_1}{N}\Bigr)
     \sin\pi k\Bigl(\half+\ffrac{j_2}{N}\Bigr)\eta^+_{j_1}\eta^-_{j_2},
\end{equation}
where (and in what follows) we assume
\begin{equation*}
 \left.\ffrac{\sin\pi j/N}{j}\right|_{j=0}=\ffrac{\pi}{N}.
\end{equation*}
Meanwhile,
the generators
$\W^{\alpha}_j$~\eqref{wbp}-\eqref{wbm} of the algebra $\Wlatq{i}{N}$ can be written as
\begin{align}\label{Whatp}
 \W^+_k&= \ffrac{2}{N}
    \sum_{j_1,j_2=-\frac{N}{2}+1}^{\frac{N}{2}-1}
     \sqrt{\ffrac{\sin\pi j_1/N}{j_1}}\sqrt{\ffrac{\sin\pi j_2/N}{j_2}} V_k(j_1,j_2) \eta^+_{j_1}\eta^+_{j_2},\\
\W^0_k&= \ffrac{2}{N}
    \sum_{j_1,j_2=-\frac{N}{2}+1}^{\frac{N}{2}-1}
     \sqrt{\ffrac{\sin\pi j_1/N}{j_1}}\sqrt{\ffrac{\sin\pi j_2/N}{j_2}} V_k(j_1,j_2)\bigl(\eta^+_{j_1}\eta^-_{j_2}+\eta^-_{j_1}\eta^+_{j_2}\bigr),\\
\label{Whatm} \W^-_k&= \ffrac{2}{N}
    \sum_{j_1,j_2=-\frac{N}{2}+1}^{\frac{N}{2}-1}
     \sqrt{\ffrac{\sin\pi j_1/N}{j_1}}\sqrt{\ffrac{\sin\pi j_2/N}{j_2}} V_k(j_1,j_2) \eta^-_{j_1}\eta^-_{j_2},
\end{align}
where we introduce trigonometric functions
\begin{equation*}
 V_k(j_1,j_2)=\sin\pi\ffrac{j_1 - j_2}{2N} \cos\pi\ffrac{(j_1 + j_2) (k + \half)}{N}
- (-1)^k\cos\pi\ffrac{j_1 + j_2}{2N}\sin\pi\ffrac{ (j_1 - j_2) (k + \half)}{N}.
\end{equation*}
We note that these functions have the symmetries
\begin{equation*}
 V_k(-j_1,-j_2)=V_k(j_2,j_1)=-V_k(j_1,j_2).
\end{equation*}

\subsection{Generalized  lattice modes}\label{sec:gen-hamilt}

Our ultimate purpose is to establish a connection between the lattice algebras and the ones present in the corresponding conformal field theory. The intuition guiding us is that the Temperley--Lieb generators are in some sense a lattice version of the stress energy tensor $T(z)$, and similarly the $\W^\alpha_j$ generators a lattice version of the triplet chiral  W-algebra currents $W^\alpha(z)$. The strategy to make this more precise is to define Fourier modes expected to converge in a certain sense (to be discussed below) towards the Virasoro $L_m$ and W-algebra modes $W^\alpha_n$.  If such a convergence occurs, one might expect that the lattice modes obey, in the limit $N\to\infty$, the commutation relations of the CFT.

Of course, the algebra of lattice modes will not in general be closed, and calculating commutators will generate more quantities on the right hand side which should also converge to the Virasoro and W-algebra modes. This means that there must be a most likely infinite family of lattice versions of the $L_m$ and $W^\alpha_n$, with convergence to a single Virasoro and W-algebra occurring in the limit $N\to\infty$. In this subsection, we define this family explicitely, postponing the discussion of what happens as $N\to\infty$ to the next section.

\subsubsection{Higher approximation of Virasoro modes}

The fact that many lattice expressions can converge to the same continuum limit is well known. For instance, there is an infinite number of choices of natural Hamiltonians for the XXZ spin-chains coming out of the Quantum Inverse Scattering formalism which all go over to the Virasoro dilatation operator $L_0$ in the continuum limit. The underlying construction \cite{KooSaleur} suggests a natural extension to the case of  non-zero modes -- that is Fourier modes at non vanishing momentum.

We start by defining the operators
\begin{equation}\label{H0}
H^0_n=\sum_{j=1}^{N-1}\cos\pi\ffrac{nj}{N} e_{j},\qquad n\in\oZ.
\end{equation}
The Hamiltonian~(\ref{XXH}) in this notation is $H^0_0$.
The adjoint action by the Hamiltonian $H^0_0$ on $H^0_n$ generates the family $H^r_n$ with $r\in\oN_0$ and $n\in\oZ$ in the following way
\begin{equation}\label{recurrentH}
 [H^0_0,H^r_m]=-4\sin\pi\ffrac{m}{2N}H^{r+1}_{m},\qquad r=0,1,2,\dots.
\end{equation}
We note that this recurrence  does not determine $H^r_0$ directly but neatly means that $H^0_0$ commutes with $H^r_0$.
To define $H^r_0$ we calculate $H^r_m$ with $m\neq0$ using the recurrence and then set $m=0$ in the final formula.
Relation~\eqref{recurrentH}
 bears some similarity with the definition of ladder operators in integrable XXZ spin chains \cite{Araki}, but is different.

We recall then the fermionic expression for the Hamiltonian in~\eqref{XXH-eta}.
By a direct calculation using the formulas~\eqref{eta-TL}, and~\eqref{H0}, and the recurrence relation~\eqref{recurrentH},
we obtain fermionic expressions for all the higher modes
\begin{multline}\label{hatH}
 H^r_n=\sum_{j_1,j_2=-\frac{N}{2}+1}^{\frac{N}{2}-1}
     \sqrt{\ffrac{\sin\pi j_1/N}{j_1}}\sqrt{\ffrac{\sin\pi j_2/N}{j_2}}\times\\
\times\Bigl(\cos^r\pi\ffrac{j_1 - j_2}{2 N} \bigl((-1)^r\delta_{j_1 + j_2,n} + \delta_{j_2 + j_1, -n}\bigr)-\\
- \sin^r\pi\ffrac{j_1 + j_2}{2 N} \bigl( \delta_{j_1-j_2, N + n}+\delta_{j_1-j_2, -N - n}
+(-1)^r\delta_{j_1 - j_2,N - n} + (-1)^r\delta_{j_1-j_2, -N + n} \bigr)\Bigr)
  \eta^+_{j_1}\eta^-_{j_2}.
\end{multline}
That these operators  give in the scaling (or continuum) limit the Virasoro modes and more generally elements in the universal enveloping of the Virasoro algebra
at central charge $c=-2$ will be discussed below in the next section. Now, we introduce what will turn out later to be  lattice analogues of modes of the triplet W-algebra currents.

\subsubsection{Higher approximations of W-algebra modes}
 Similarly, we  first define the operators
\begin{equation}
 W^{\alpha,0}_n=\sum_{j=1}^{N-2}\cos\pi\ffrac{n(j+\frac{1}{2})}{N}\W^\alpha_{j},
\qquad\alpha=+,-,0.
\end{equation}
The adjoint action by the Hamiltonian $H^0_0$ on $W^{\alpha,0}_n$ generates the family $W^{\alpha,r}_n$ with $r\in\oN_0$ and $n\in\oZ$ in the following way
\begin{equation}\label{recurrentW}
 [H^0_0,W^{\alpha,r}_m]=-4\sin\pi\ffrac{m}{2N}W^{\alpha,r+1}_{m},\qquad r=0,1,2,\dots,\qquad\alpha=+,-,0.
\end{equation}
We note that $W^{\alpha,r}_0$ is determined in the same way as it is explained after~(\ref{recurrentH}).

Using~\eqref{XXH-eta} for $H^0_0=H$ and~\eqref{Whatp}-\eqref{Whatm} with the recurrent relation~\eqref{recurrentW},
we obtain
\begin{align}\label{hatWp}
 W^{+,r}_n&= \sum_{j_1,j_2=-\frac{N}{2}+1}^{\frac{N}{2}-1}
     \sqrt{\ffrac{\sin\pi j_1/N}{j_1}}\sqrt{\ffrac{\sin\pi j_2/N}{j_2}} U_{r,n}(j_1,j_2) \eta^+_{j_1}\eta^+_{j_2},\\
W^{0,r}_n&= \sum_{j_1,j_2=-\frac{N}{2}+1}^{\frac{N}{2}-1}
     \sqrt{\ffrac{\sin\pi j_1/N}{j_1}}\sqrt{\ffrac{\sin\pi j_2/N}{j_2}} U_{r,n}(j_1,j_2)
\Bigl(\eta^+_{j_1}\eta^-_{j_2}+\eta^-_{j_1}\eta^+_{j_2}\Bigr),\label{hatWz}\\
W^{-,r}_n&= \sum_{j_1,j_2=-\frac{N}{2}+1}^{\frac{N}{2}-1}
     \sqrt{\ffrac{\sin\pi j_1/N}{j_1}}\sqrt{\ffrac{\sin\pi j_2/N}{j_2}} U_{r,n}(j_1,j_2) \eta^-_{j_1}\eta^-_{j_2},
\label{hatWm}
\end{align}
where we introduce trigonometric functions
\begin{multline}
  U_{r,n}(j_1,j_2)= \sin\pi\ffrac{j_1 - j_2}{2 N}\cos^r\pi\ffrac{j_1 - j_2}{2 N}
       \bigl((-1)^r\delta_{j_1 + j_2, n}+\delta_{j_1 + j_2, -n}\bigr)\\
+2 \cos\pi\ffrac{j_1 + j_2}{2 N}\sin^r\pi\ffrac{j_1 + j_2}{2 N} \bigl(\delta_{j_1 - j_2, -N-n} +(-1)^r
      \delta_{j_1 - j_2, -N+n}\bigr).
\end{multline}

\subsection{A note on the odd number of sites}\label{sec:odd-N-1}
We can also consider the spin-chains with odd number~$N$ of sites. They
will give below a different sector of symplectic fermions in the
scaling limit.
For the odd-$N$, we define the fermions $\eta^\pm_n$ with half-integer $-\frac{N}{2}+1\leq n\leq \frac{N}{2}-1$ by
formulas~\eqref{etapm} where the fermionic operators $\theta_k$ and $\theta^{\dagger}_k$, with $0\leq k\leq N-1$, are defined by~\eqref{theta-ferm}
with amplitudes $A_k(j)$ defined by~\eqref{amplitudes}, for $k=1,\dots, N-1$,
and $A_0(j)=i^j$. We also have the $\UresSL2$
generators (\ref{QG-ferm-1}) represented as $\E=\theta^{\dagger}_0 K$ and $\F=\theta_0$ satisfying
\begin{equation}
 \{\E\K^{-1},\eta^\pm_j\}=0,\quad\{\F,\eta^\pm_j\}=0,\qquad j=-\ffrac{N}{2}+1,-\ffrac{N}{2}+2,\dots,\ffrac{N}{2}-1,
\end{equation}
in addition to the quantum group relations.
Then the backward transformation is
\begin{align}
 c^\dagger_k&= -i \sum_{j=-N/2+1}^{N/2-1}
     \frac{1}{j} \sqrt{\ffrac{j}{N\sin\pi j/N}} A_{N/2 + j}( k)\eta^+_j-  i^{k}\E\K^{-1},\\
c_k&=-  \sum_{j=-N/2+1}^{N/2-1}  \frac{1}{j} \sqrt{\ffrac{j}{N\sin\pi j/N}} A_{N/2 - j}( k)\eta^-_j + i^{k-1}\F.
\end{align}
The $s\ell(2)$ generators are now represented  as
\begin{gather*}
 \e=\sum_{n=1/2}^{N/2-1}\ffrac{1}{n}\eta^+_{-n}\eta^+_n,\qquad
 \f=-\sum_{n=1/2}^{N/2-1}\ffrac{1}{n}\eta^-_{-n}\eta^-_n,\\
\h=-\ffrac{1}{2}\sum_{n=1/2}^{N/2-1}\ffrac{1}{n}
\bigl(\eta^+_{-n}\eta^-_n+\eta^-_{-n}\eta^+_n\bigr).
\end{gather*}

We note that the Hamiltonian (\ref{XXH}) takes a slightly different form
\begin{equation}\label{XXH-eta-odd}
 H=\sum_{n=1/2}^{N/2-1}\ffrac{2}{n}\sin\pi\ffrac{n}{N}\bigl(\eta^-_{-n}\eta^+_{n}
-\eta^+_{-n}\eta^-_n\bigr)
+\Bigl(1-\ffrac{1}{\sin\frac{\pi}{2N}}\Bigr).
\end{equation}

The definition of the lattice W-algebra is the same as in the even-$N$ case and Thm.~\bref{thm:centralizer} is true for odd-$N$ case as well.
Formal expressions of the lattice W-algebra generators in terms of $\eta^{\pm}$ fermions are the same as in the even-$N$ case.
So, the generators $e_k$ are given by~\eqref{eta-TL} and $W^{\alpha}_k$ are given by (\ref{Whatp})-(\ref{Whatm}).
Similarly, their linear combinations $H^0_n$ and $W^{\alpha,0}_n$ together with higher ``currents'' $H^r_n$ and $W^{\alpha,r}_n$ are expressed in the same way. So, the currents $H^r_n$ are given by (\ref{hatH}) and currents $W^{\alpha,r}_n$ by (\ref{hatWp})--(\ref{hatWm}).
For odd $N$ the operators $H^r_n$ and $W^{\alpha,r}_n$ satisfy the same commutation relations as for even $N$.

\medskip

Next, we turn to the study of the scaling properties of the whole
family of the operators (modes) $H^r_n$ and $W^{\alpha,r}_n$
introduced in~\bref{sec:gen-hamilt}.

\section{Scaling limit of the XX chain and W-algebra modes\label{sec:scal-lim}}
In this section we calculate the limit $N\to\infty$ of the commutation relations
between modes (\ref{hatH}) and (\ref{hatWp})-(\ref{hatWm}) and show that they
coincide with the commutation relations of the triplet W-algebra modes.
For the convenience of the reader we give all the necessary conventions about symplectic
fermions and the triplet W-algebra in App.~\bref{App:W-tripl}.

\subsection{Scaling limit and $1/N$-decomposition}
Here, we discuss how to proceed from the $\Wlatq{i}{N}$
generators to get all the modes of the triplet $W$-algebra currents, including the Virasoro generators, in the chiral logarithmic
conformal field theory of symplectic fermions. The procedure is called the scaling limit $N\to\infty$ of the model. This limit involves the rescaled Hamiltonian $NH$ and its ``low-lying'' eigenvectors: one should choose the vacuum state $|0\rangle$  and consider in the limit only finite-energy states for $N(H-E_0)$, with the ground state energy $E_0=\langle0|H|0\rangle$, or in our case the states generated by $\eta^{\pm}_n$ and $\gamma^{\pm}$ from $|0\rangle$ for any finite $n$. This way we obtain the CFT Hilbert space and then we should take the limit of the higher modes $H_n^r$ and $W_n^{\alpha,r}$ rescaled also by $N$ in the basis of the ``low-lying'' eigenvectors. For these reasons, we found the decomposition of  the Fourier modes $H_n^r$ and $W_n^{\alpha,r}$ in Sec.~\bref{sec:gen-hamilt} in terms of operators that generate eigenvectors from the vacuum of $H$ --  the fermionic bilinears $\eta^{\alpha}_j\eta^{\beta}_k$ (note that coefficients in the decompositions are functions of the spin-chain size $N$.) We will show below that $H_n^r$ and $W_n^{\alpha,r}$ generate a Lie algebra whose structure constants are trigonometric  functions of $N$. The scaling limit of the Fourier modes $H_n^r$ and $W_n^{\alpha,r}$  and their commutators involves thus formal expansions of the operators and the structure constants, respectively, as Laurent series in $N$. Different terms in such expansions for $H_n^r$ and $W_n^{\alpha,r}$ are then indeed identified with
 elements from the enveloping algebra of the chiral triplet $W$-algebra. Similar constructions of limits were discussed in~\cite{GRS1,GRS3} where also more formal construction of direct/inductive limits is introduced.

\subsubsection{The scaling limit of the Hamiltonian}
We first study the scaling limit of the
Hamiltonian~$H$. We take the Hamiltonian~\eqref{XXH-eta}, which
is written in the normal-ordered form and keep terms up to the order $1/N$
for the very large $N$ decomposition
\begin{equation}
 H=-2\ffrac{N}{\pi}+1+2\ffrac{\pi}{N}\Bigl(-\ffrac{(-2)}{24}+\sum_{n=1}^{N/2-1}\bigl(\eta^-_{-n}\eta^+_{n}
-\eta^+_{-n}\eta^-_n\bigr)
-\eta^+_0\eta^-_0\Bigr)+O(1/N^2).
\end{equation}
We see that the leading part (after subtracting the vacuum expectation of the Hamiltonian) which is an operator in front of $1/N$ has at $N\to\infty$ the form  $L_0-c/24$ on the ``low-lying'' eigenvectors or scaling states, with $c=-2$. This expression of $L_0$ is well-known and it is the zero mode of the stress-energy tensor $T(z)$ in the symplectic fermions theory (compare it with~\eqref{L-modes}). Thus, properly renormalized Hamiltonian $\frac{N}{2\pi}(H-\langle0|H|0\rangle)$ converges to $L_0-c/24$.

\subsubsection{Higher Hamiltonians}
We could expect new features considering the scaling limit of the
whole family of the higher Hamiltonians $H^r_0$ of the open XX spin-chain
introduced above.

The limit $N\to\infty$ of operators $H^{2k}_0$ given by (\ref{hatH}) is singular and therefore we should subtract
the divergent part, as we did it for $H=H^0_0$. To do this we calculate the eigenvalues $2h_\ell$ of $H^{\ell}_0$ on the vacuum vector~(\ref{vacuum})
using~(\ref{hatH}), which gives
\begin{gather}\label{vac-eigen}
 h_0=\half\bigl(\cot\ffrac{\pi}{2N}-1\bigr),\\
\notag h_\ell=-\ffrac{1}{2(\ell+1)}+\ffrac{1}{2^{\ell+1}}
\sum_{n=0}^{\ell/2}\left(\left({\footnotesize\ell-1\atop\ell/2+n-1}\right)-\left({\footnotesize\ell-1\atop\ell/2+n+1}\right)\right)
\cot\pi\ffrac{2n+1}{2N},\quad\mbox{\rm for even $\ell>0$.}
\end{gather}
As we will see from general formulas below, the first leading term (in $1/N$) of the renormalized Hamiltonians
$\frac{N}{2\pi}(H^l_0-2h_l)$ gives $L_0$ again but the next subleading terms, those in front of $1/N^k$, are different
from those in the $H^0_0$ decomposition and are particular polynomials in zero modes of $T(z)$ and its descendants.

\subsubsection{Virasoro modes in symplectic fermions}
We start from calculating the scaling limit of the Fourier modes $H^0_n$.
Expanding $\sqrt{\sin\pi j/N}$ in (\ref{hatH}) into a series in $\frac{1}{N}$ we obtain the very large $N$ decomposition
\begin{multline*}
 \frac{N}{\pi}(H^0_n-2h_0\delta_{n,0})=-\sum_{j=-\frac{N}{2}+1}^{\frac{N}{2}-1}:\eta^+_{j}\eta^-_{n-j}:
                                       -\sum_{j=-\frac{N}{2}+1}^{\frac{N}{2}-1}:\eta^+_{j}\eta^-_{-n-j}:+\\
+\sum_{j_1=-\frac{N}{2}+1}^{\frac{N}{2}-1}\sum_{j_2=-\frac{N}{2}+1}^{\frac{N}{2}-1}\bigl( \delta_{j_1-j_2, N + n}+\delta_{j_1-j_2, -N - n}
+\delta_{j_1 - j_2,N - n} +\delta_{j_1-j_2, -N + n} \bigr)
  \eta^+_{j_1}\eta^-_{j_2}+
o\bigl(\ffrac{1}{N}\bigr),
\end{multline*}
where we use the standard normal ordering
\begin{equation}
:\eta^{\alpha}_{j_1}\eta^{\beta}_{j_2}: =
\begin{cases}
\eta^{\alpha}_{j_1}\eta^{\beta}_{j_2},&\qquad j_2\geq0,\\
-\eta^{\beta}_{j_2}\eta^{\alpha}_{j_1},&\qquad j_1\geq0.
\end{cases}
\end{equation}
We note that the second line of the previous expression vanishes on the scaling or the low-energy states because for any $n$ we can find $N$ such that $j_1$ or $j_2$
will be larger than any fixed number. Therefore, on the scaling states we have
\begin{equation}
 \ffrac{N}{\pi}(H^0_n-2h_0\delta_{n,0})=L^{(N)}_n+L^{(N)}_{-n} + o\bigl(\ffrac{1}{N}\bigr),
\end{equation}
where
\begin{equation}
 L^{(N)}_n=-\sum_{j=-\frac{N}{2}+1}^{\frac{N}{2}-1}:\eta^+_{j}\eta^-_{n-j}:.
\end{equation}
This expression coinsides with (\ref{L-modes}) in the limit $N\to\infty$. In what follows we suppress
the superscript ${}^{(N)}$ in the expressions whenever it cannot lead to a confusion.

Similarly, after subtracting the vacuum eigenvalues we can write the series decomposition for the whole family of higher modes $H^r_n$ in $1/N$
on the scaling states:
\begin{equation}\label{H-decomp-N}
 H^r_n-(1+(-1)^r)h_r\delta_{n,0}=\sum_{k=1}^\infty\left(\ffrac{\pi}{N}\right)^{2k-1}H^r_n[k],
\end{equation}
where first several leading and subleading terms in front of $1/N^k$ have at $N\to\infty$ the form
\begin{align}
\label{H1} H^r_n[1]&=(-1)^rL_n+L_{-n},\\
 H^r_n[2]&=\ffrac{1}{24}\Bigl(4 (3 r+1)\bigl((-1)^r(L^2)_n+(L^2)_{-n}\bigr) - 3 r \bigl((-1)^r(\partial^2L)_n+(\partial^2L)_{-n}\bigr)+\\
\notag&\kern100pt+ 3(r+2) \bigl((-1)^r(\partial L)_n+(\partial L)_{-n}\bigr) + 4 \bigl((-1)^rL_n+L_{-n}\bigr)\Bigr),\\
H^r_n[3]&=\ffrac{1}{5760}\Bigl(-64 (15 r^2+1)\bigl((-1)^r(L^3)_n+(L^3)_{-n}\bigr) + \\
\notag&\kern20pt+ 8 (2 + 15 r( r-1))\bigl((-1)^r(\partial^2L^2)_n+(\partial^2L^2)_{-n}\bigr) +\\
\notag&\kern40pt+ (-4 + 15 r ( r+2))\bigl((-1)^r(\partial^4L)_n+(\partial^4L)_{-n}\bigr) - \\
 \notag&\kern60pt-1080 r (r+1) \bigl((-1)^r(\partial^2L^2)_n+(\partial^2L^2)_{-n}\bigr)+ \\
 \notag&\kern80pt
+ 10 (-4 + 3 r (9 r+8))\bigl((-1)^r(\partial^3L)_n+(\partial^3L)_{-n}\bigr) - \\
 \notag&\kern100pt
-  80 (9 r^2+ 24 r-1) \bigl((-1)^r(L^2)_n+(L^2)_{-n}\bigr) + \\
 \notag&\kern120pt
+ 45 ( r^2+ 6 r-4 ) \bigl((-1)^r(\partial^2L)_n+(\partial^2L)_{-n}\bigr) -  \\
 \notag&\kern140pt
 -45 (r^2+ 6 r+4) \bigl((-1)^r(\partial L)_n+(\partial L)_{-n}\bigr) -\\
 \notag&\kern220pt - 48 \bigl((-1)^rL_n+L_{-n}\bigr)\Bigr),
\end{align}
where  the modes $(\partial^k L^m)_n$ of the composite fields $:\!\partial^k T(z)^m\!:$\, are computed using the fermionic expression~\eqref{eq:T} of the stress energy tensor $T(z)$, expressions for $L_n$  and $(L^2)_n$  given
in~\eqref{L-modes}-\eqref{L2-modes}. Note that for $n=0$ we obtain in this way an
infinite family of commuting Hamiltonians $[H^r_0,H^k_0]=0$ as  particular polynomials in zero modes of $T(z)$ and its descendants.

In order to obtain the defining relations of the Virasoro algebra in the continuum limit,
only the leading $1/N$ term will be necessary below, that is the $H^r_n[1]$.
For the triplet W-algebra relations, the presence of non linearities requires
keeping also $H^r_n[2]$, and similarly (see below) $W^{\alpha,r}_n[1]$ and $W^{\alpha,r}_n[2]$.

\subsubsection{W-algebra modes in symplectic fermions}
In a similar way, we obtain the decompositions of (\ref{hatWp})--(\ref{hatWm})
\begin{equation}\label{W-decomp}
 W^{\alpha,r}_n=\sum_{k=1}^\infty\left(\ffrac{\pi}{N}\right)^{2k}W^{\alpha,r}_n[k],
\end{equation}
where the first two coefficients have at $N\to\infty$ the form
\begin{align}\label{W-decomp-coefs}
 W^{\alpha,r}_n[1]&=(-1)^rW^\alpha_n+ W^\alpha_{-n},\\
 W^{\alpha,r}_n[2]&=\ffrac{1}{48}\Bigl(\ffrac{48}{5} ( 3 r+2)\bigl((-1)^r(LW^\alpha)_n+(LW^\alpha)_{-n}\bigr)-\\
\notag&\kern60pt
-\ffrac{2}{5} (9 r + 1)\bigl((-1)^r(\partial^2W^\alpha)_n+(\partial^2W^\alpha)_{-n}\bigr)-\\
\notag&\kern80pt
-(18 r + 26)\bigl((-1)^r(\partial W^\alpha)_n+(\partial W^\alpha)_{-n}\bigr)-\\
\notag&\kern100pt
-(6 r+22)\bigl((-1)^r W^\alpha_n+ W^\alpha_{-n}\bigr)\Bigr),
\end{align}
with expressions for $W^{\alpha}_n$  and $(LW^{\alpha})_n$  given
in~\eqref{Wpm-modes}-\eqref{LW-modes}. Note that the leading term in $W^{\alpha,r}_n$
appears at $1/N^2$ and we have a convergence of the rescaled Fourier modes $\frac{N^2}{\pi^2} W^{\alpha,r}_n$ to
the modes $(-1)^rW^\alpha_n+ W^\alpha_{-n}$ of the chiral W-algebra currents, for any integer~$n$.

\subsection{Scaling limit of the commutators}
In this section we calculate commutators of $H^r_n$ and $W^{\alpha,r}_n$ and show that in the limit $N\to\infty$ the commutators
give the commutation relations (\ref{LL})--(\ref{WW}) of the chiral triplet W-algebra of symplectic fermions.
Using the explicit expression (\ref{hatH}) we obtain the following commutation relations
\begin{align}\label{H0H0}
 [H^0_n,H^0_m]&=2\sin\pi\ffrac{n-m}{2N}H^1_{n+m}+
  2\sin\pi\ffrac{n+m}{2N}H^1_{n-m},\\
 [H^0_n,H^1_m]&=2\sin\pi\ffrac{2n-m}{2N}H^2_{n+m}-
  2\sin\pi\ffrac{2n+m}{2N}H^2_{n-m}-\\
&\notag\kern50pt-2\sin\pi\ffrac{n}{2N}\bigl(\cos\pi\ffrac{n-m}{2N}H^0_{n+m}-\cos\pi\ffrac{n+m}{2N}H^0_{n-m}\bigr),\\
\label{H1H1} [H^1_n,H^1_m]&=2\sin\pi\ffrac{n-m}{N}H^3_{n+m}
-2\sin\pi\ffrac{n+m}{N}H^3_{n-m}-\\
&\notag\kern50pt-2\sin\pi\ffrac{n-m}{2N}(\cos\pi\ffrac{n+m}{2N}+\sin\pi\ffrac{n}{2N}\sin\pi\ffrac{m}{2N})H^1_{n+m}+\\
&\notag\kern90pt+2\sin\pi\ffrac{n+m}{2N}(\cos\pi\ffrac{n-m}{2N}-\sin\pi\ffrac{n}{2N}\sin\pi\ffrac{m}{2N})H^1_{n-m}.
\end{align}
Taking (\ref{H-decomp-N}) and the sign factor $(-1)^r$ in (\ref{H1}) into account we obtain the commutation relations of $L_n$'s as the limit
\begin{equation}\label{LL-N}
 [L_n,L_m]=\lim_{N\to\infty}\ffrac{N^2}{4\pi^2}\bigl([H^0_n,H^0_m]-[H^0_n,H^1_m]+[H^0_m,H^1_n]+[H^1_n,H^1_m]\bigr).
\end{equation}
Then, we substitute right hand sides of (\ref{H0H0})--(\ref{H1H1}), then substitute $H^r_n$ from (\ref{H-decomp-N})
and calculate the Laurent series in $N$ up to the second order which indeed gives the Virasoro commutation relations~(\ref{LL}).
We note that in this calculation we need to keep only first term in the decomposition~(\ref{H-decomp-N}). We also note
that the Virasoro central extension appears from the decomposition of the vacuum eigenvalues~(\ref{vac-eigen}) in $N$.

\medskip

The commutation relations $[H^r_n,W^{\alpha,s}_m]$ and $[W^{\alpha,r}_n,W^{\beta,s}_m]$ are given in App.~\bref{sec:comm-rel}.
To calculate the commutators
\begin{equation}
 [L_n,W^{\alpha}_m]=\lim_{N\to\infty}\ffrac{N^3}{4\pi^3}\bigl([H^0_n,W^{\alpha,0}_m]-[H^0_n,W^{\alpha,1}_m]-[H^1_m,W^{\alpha,0}_m]
+[H^1_n,W^{\alpha,1}_m]\bigr)
\end{equation}
we take the commutators from (\ref{H0W0})--(\ref{H1W1}) and substitute the Laurent decompositions (\ref{H-decomp-N}) and (\ref{W-decomp})
into them.
This gives the relations~(\ref{LW}).

\medskip

The calculation of the commutator
\begin{equation}
 [W^{\alpha}_n,W^{\beta}_m]=\lim_{N\to\infty}\ffrac{N^4}{4\pi^4}\bigl([W^{\alpha,0}_n,W^{\beta,0}_m]-[W^{\alpha,1}_n,W^{\beta,1}_m]
-[W^{\alpha,1}_m,W^{\beta,0}_m]+[W^{\alpha,1}_n,W^{\beta,1}_m]\bigr)
\end{equation}
differs from the previous two in the following point. Due to the factor $N^4$ 
 we should take the decompositions
of the right hand sides of (\ref{W0W0})--(\ref{W1W1}) up to the $4$th order. The leading order coming from the trigonometric functions is $0$ or $1$
and therefore we should keep the  $3$rd and $4$th orders in the decompositions (\ref{H-decomp-N}) and (\ref{W-decomp}) respectively.
This calculation results in the contribution from $(L^2)_n$ and $(LW^\alpha)_n$ in the chiral  W-algebra relations~(\ref{WW}), which is obtained after taking the limit. We have thus succeeded
in reproducing the algebra structure of the Virasoro $L_n$ and the $W^{\alpha}_n$ modes from our lattice approximations.

\subsection{A note on the odd-$N$ case scaling limit}\label{sec:odd-N-2}
The calculation of the $N\to\infty$ limit for the odd $N$ differs from the calculation
for the even $N$ only by normal ordering constants. Thus for the odd case the calculation
reapeats the even case but now instead of (\ref{vac-eigen})
we take the normal ordering constants
\begin{gather}\label{vac-eigen-odd}
 h_0=\half\Bigl(\ffrac{1}{\sin\frac{\pi}{2N}}-1\Bigr),\\
\notag h_\ell=-\ffrac{1}{2(\ell+1)}+\ffrac{1}{2^{\ell+1}}
\sum_{n=0}^{\ell/2}\left(\left({\footnotesize\ell-1\atop\ell/2+n-1}\right)-\left({\footnotesize\ell-1\atop\ell/2+n+1}\right)\right)
\ffrac{1}{\sin\pi\frac{2n+1}{2N}},\quad\mbox{\rm for even $\ell>0$}.
\end{gather}

\section{The XXZ representation and the W-algebra on a lattice}\label{sec:XXZ}
In this section, we give our definition of the lattice W-algebra in
the context of XXZ spin-chains at any root of unity $\q$. We begin with collecting basic facts about XXZ spin-chains with
quantum-group symmetry~\cite{PasquierSaleur,ReadSaleur07-2,GV}. Then, we define the lattice $W$-algebra using its centralizing property with the small quantum group and study its properties (generators, relations, representation theory, etc.) Finally, in the last subsection we discuss two scaling limits of this lattice algebra, describing $(p-1,p)$ and $(1,p)$ models.

\subsection{The Hamiltonian}
We consider the XXZ Hamiltonian of $N$ one-half spins with an open
boundary condition described by the ``quantum-group symmetric'' boundary term~\cite{PasquierSaleur},
\begin{equation}\label{XXZ-sigma}
H(\q) = \half\sum_{i=1}^{N-1}(\sigmax_i \sigmax_{i+1} + \sigmay_i \sigmay_{i+1}
+ \ffrac{\q+\q^{-1}}{2} \sigmaz_i \sigmaz_{i+1}) - \ffrac{\q-\q^{-1}}{4}(\sigmaz_1-\sigmaz_{N}),
\end{equation}
where $\sigmax_i, \sigmay_i$ and $\sigmaz_i$ are usual Pauli matrices
given in~\eqref{Pauli}.
We note that we have an opposite sign in front of the boundary
term~\eqref{XXZ-sigma}, with respect to the standard
choice~\cite{PasquierSaleur,ReadSaleur07-2}, which corresponds to a slightly different basis (corresponding
to $\q\to\q^{-1}$ or relabeling of sites from $N$ to $1$) we use in
order to adapt to our quantum-group notations.

 The Hamiltonian $H(\q)$ is not
hermitian (with respect to the usual bilinear form on the spin chain) due to the imaginary boundary term
$\frac{\sin(\pi/p)}{2\rmi}$ but its eigenvalues are real while its
Jordan form has cells of rank $2$. An easiest way to see non-trivial Jordan
forms in the spectrum is to study symmetries of the Hamiltonian densities which we discuss below.

We note finally that the Hamiltonian is massless when $|\q|=1$ and we
will consider only this critical case. Moreover,
roots of unity cases are most interesting for applications in the
scaling limit, and these are the ones we will consider  for the definition and study
of lattice versions of the triplet W-algebras.
We recall some  basic information about the quantum group $\LUresSL2$
symmetry at the root of unity cases in App.~\bref{app:qunatm-gr-def} where we also give relations to
more usual (in the spin-chain literature) quantum group generators
$S^{\pm}$, $S^z$ and $\q^{S^z}$.

\subsection{The Casimir operator and Temperley--Lieb algebra}
As a module over $\LQG$, the XXZ spin
chain $\chVv$ is a tensor product of $N$ copies of the
fundamental two-dimensional simple module $\XX_{2,1}$ (see
notations below in~\bref{sec:irreps})
 such that  the genereators are represented on $\XX_{2,1}$ as
\begin{equation*}
\E = \sigmap, \quad \F = \sigmam,  \quad \K = \ffrac{\q+\q^{-1}}{2}\one + \ffrac{\q-\q^{-1}}{2}\sigmaz, \quad \text{and}\quad \e=\f=\h=0.
\end{equation*}
Using the $(N-1)$-folded
comultiplications~\eqref{N-fold-comult-cap} for the capital generators together with   the general comultiplication~\eqref{N-fold-comult-e} and~\eqref{N-fold-comult-f} for the divided powers, which we  compute in App.~\bref{app:qunatm-gr-def},
we obtain
the XXZ representation of $\LQG$
\begin{equation}
\repQGq:\LQG\to\Endo_{\oC}(\chVv), \qquad\chVv = \bigotimes_{j=1}^{N}\XX_{2,1}.
\end{equation}
It has the usual XXZ expressions for the generators
\begin{equation}\label{XXZrep-EF}
\begin{split}
\repQGq(\E)&=\sum_{j=1}^{N}\one\tensor\dots\tensor\one\tensor
\sigma^+_j\tensor \q^{\sigmaz_{j+1}}\tensor \dots \tensor  \q^{\sigmaz_N},\\
\repQGq(\F) &=
\sum_{j=1}^{N} \q^{-\sigmaz_1}\tensor\dots\tensor \q^{-\sigmaz_{j-1}}\tensor
\sigma^-_j\tensor \one \tensor \dots \tensor \one
\end{split}
\end{equation}
and for the divided powers the action is
\begin{multline}\label{XXZrep-e}
\repQGq(\e)= \q^{\frac{p(p-1)}{2}}\Bigl(\bigotimes_{j=1}^{{N}}\sigmaz_j\Bigr)
 \sum_{1\leq j_1<j_2<\dots<j_p\leq{N}}(-1)^{j_p}
    \one\tensor\dots\tensor\one\tensor\sigmap_{j_1}\tensor\q^{\sigmaz_{j_1+1}}\tensor\q^{\sigmaz_{j_1+2}}\tensor\dots\\
  \dots\tensor\q^{(k-1)\sigmaz_{j_k-1}}\tensor\sigmap_{j_k}\q^{(k-1)\sigmaz_{j_k}}
   \tensor\q^{k\sigmaz_{j_k+1}}\tensor\dots\tensor\sigmap_{j_p}\q^{(p-1)\sigmaz_{j_p}}\tensor\sigmaz_{j_p+1}\tensor\dots\tensor\sigmaz_{N}
   \end{multline}
 and
 \begin{multline}\label{XXZrep-f}
\repQGq(\f)=(-1)^{p}\Bigl(\bigotimes_{j=1}^{{N}}\sigmaz_j\Bigr)
   \sum_{1\leq j_1<j_2<\dots<j_p\leq{N}}(-1)^{j_p}
   \one\tensor\dots\tensor\one\tensor\sigmam_{j_1}\tensor\q^{\sigmaz_{j_1+1}}\tensor\q^{\sigmaz_{j_1+2}}\tensor\dots\\
  \dots\tensor\q^{(k-1)\sigmaz_{j_k-1}}\tensor\sigmam_{j_k}
   \tensor\q^{k\sigmaz_{j_k+1}}\tensor\dots\tensor\sigmam_{j_p}\tensor\sigmaz_{j_p+1}\tensor\dots\tensor\sigmaz_{N}.
\end{multline}

Then, we recall that the quantum-group Casimir element
\begin{equation}\label{eq:casimir}
  \cas=
  \F\E+\ffrac{\q \K+\q^{-1}\K^{-1}}{(\q-\q^{-1})^2}
\end{equation}
commutes with the full  quantum
group $\LQG$ and
it is represented  on $\XX_{2,1}\tensor\XX_{2,1}$ as
\begin{multline*}
\Delta(\cas) = \half\bigl( \sigmax\tensor\sigmax + \sigmay\tensor\sigmay
	+ \ffrac{\q + \q^{-1}}{2}\sigmaz\tensor\sigmaz\bigr)\\
	+ \ffrac{\q - \q^{-1}}{4}\bigl(\one\tensor\sigmaz - \sigmaz\tensor\one\bigr)
	+\ffrac{3\q^3+ \q + \q^{-1} + 3\q^{-3}}{4(\q-\q^{-1})^2}\one\tensor\one.
\end{multline*}
This gives the usual XXZ coupling and the
Hamiltonian~\eqref{XXZ-sigma} is expressed as a sum over the densities
$\Delta(\cas)_{i,i+1}$ acting on $i$th and $(i+1)$th tensorands:
\begin{equation*}\label{XXZ-C}
H(\q)=\sum_{i=1}^{N-1}\Delta(\cas)_{i,i+1} - (N-1)\ffrac{3\q^3+ \q + \q^{-1} + 3\q^{-3}}{4(\q-\q^{-1})^2}.
\end{equation*}
Introducing the  more convenient notation for the Hamiltonian densities
\begin{equation}\label{cas-TL}
e_i = -\Delta(\cas)_{i,i+1}  + \ffrac{\q^3 + \q^{-3}}{(\q-\q^{-1})^2}\one,
 \qquad 1\leq i\leq N-1,
\end{equation}
the Hamiltonian takes the usual form~\cite{PasquierSaleur,ReadSaleur07-2}
\begin{equation}\label{XXZ-e}
H(\q)=-\sum_{i=1}^{N-1}e_i + (N-1)\ffrac{\q+\q^{-1}}{4}.
\end{equation}

It is straightforward to check that the operators $e_i$, for $1\leq i\leq N-1$,
satisfy the defining relations of the Temperley--Lieb algebra $\TLq{N}$
\begin{align*}
e_i^2 &= (\q+\q^{-1})e_i,\\
e_i e_{i\pm1}e_i &= e_i,\\
e_i e_j &= e_j e_i,\quad |i-j|>1.
\end{align*}

The relation~\eqref{cas-TL} of the Temperley--Lieb generators $e_i$ with the
 Casimir element and the coassociativity\footnote{By the coassociativity, we can write the
action of $\E$ on $\oC^2\otimes\oC^2\otimes\oC^2$ in two ways: (i) $\Delta^2(E)=\Delta(\one)\otimes\E + \Delta(\E)\otimes\K$
which obviously commutes with $e_1$ and (ii) $\Delta^2(E)=\one\otimes\Delta(\E) + \E\otimes\Delta(\K)$ which
now clearly commutes with $e_2$, and similarly for the other $\LQG$ generators and higher values of $N$.} of the
quantum-group comultiplication give the well-known result~\cite{PasquierSaleur}
\begin{equation}
\left[\repQGq\bigl(\LQG\bigr),\TLq{N}\right] = 0.
\end{equation}
Moreover, it was shown first in~\cite{M1} that  $\TLq{N}$ is the centralizer of the
representation of $\LQG$ on $\chVv$ and vice versa, i.e., they are
mutual centralizers,
\begin{equation}
\TLq{N}\cong\Endo_{\LQG}(\chVv),\qquad \repQGq\bigl(\LQG\bigr) \cong \Endo_{\TLq{N}}(\chVv),
\end{equation}
for any $\q$, including  the root of unity cases.

It was also shown~\cite{M1,ReadSaleur07-2} that the representation~\eqref{cas-TL} is faithful
and  projective $\TLq{N}$-modules
 are direct summands in a decomposition of the spin-chain $\chVv$.
Such a decomposition over $\TLq{N}$ can be systematically found using the symmetry
algebra -- the centralizing algebra $\LQG$ and a decomposition of $\chVv$ as a module over $\LQG$,
see e.g.~\cite{GV}. We will analyze these decompositions for explicitly describing generators of our new lattice W-algebra and in studying its modules below.
It is thus important to recall basics of representation theory of
the quantum groups $\LQG$ and $\UresSL2$
 for any root of unity first before going into details on the lattice W-algebra and a
 spin-chain decomposition over it.

\subsection{Modules over the quantum groups}
In this section, we describe a few important $\LQG$- and $\UresSL2$-modules which appear in  decompositions of the spin-chains. We begin with recalling results on the simple modules and then describe their projective covers~\cite{[BFGT]}.

\subsubsection{Simple $\LUresSL2$-modules}\label{sec:irreps}
A simple $\LUresSL2$-module $\repX_{s,r}$ is labeled by
the pair $(s,r)$, with $1\leq s\leq p$ and $r\in\oN$, and has the highest
weights $(-1)^{r-1}\q^{s-1}$ and $\frac{r-1}{2}$ with respect to $\K$ and $\h$
generators, respectively. The $sr$-dimensional module $\XX_{s,r}$
is spanned by elements $\stprp_{n,m}$, $0\leq n\leq s{-}1$,
$0\leq m\leq r{-}1$, where $\stprp_{0,0}$ is the highest-weight
vector and the left action of the algebra on $\XX_{s,r}$ is
given~by
\begin{align}
  \K \stprp_{n,m} &=
  (-1)^{r-1} \q^{s - 1 - 2n} \stprp_{n,m},\qquad
  &\h\, \stprp_{n,m} &=  \half(r-1-2m)\stprp_{n,m},\label{basis-lusz-irrep-1}\\
  \E \stprp_{n,m} &=
  (-1)^{r-1} [n][s - n]\stprp_{n - 1,m},\qquad
  &\e\, \stprp_{n,m} &=  m(r-m)\stprp_{n,m-1},\label{basis-lusz-irrep-2}\\
  \F \stprp_{n,m} &= \stprp_{n + 1,m},\qquad
  &\f\, \stprp_{n,m} &=  \stprp_{n,m+1},\label{basis-lusz-irrep-3}
\end{align}
where we set $\stprp_{-1,m}=\stprp_{n,-1}
=\stprp_{s,m} =\stprp_{n,r}=0$.
We note  that the  $\XX_{s,r}$ restricted to the subalgebra $\UresSL2$ is
isomorphic to the $r$-fold direct sum of the simple $\UresSL2$-module $\XX_{s}^\alpha$ with $\alpha=(-1)^{r-1}$
in notations of~\cite{[FGST]} and the  action of $\UresSL2$-generators is given by the first column of~\eqref{basis-lusz-irrep-1}-\eqref{basis-lusz-irrep-3}.

\subsubsection{Projective covers}
The simple modules $\XX_{s,r}$ have projective covers $\PP_{s,r}$
with the subquotient structure
\begin{equation}\label{proj-diag-LQG}
 \xymatrix@=12pt{
    &\stackrel{\XX_{s,1}}{\bullet}\ar[d]_{}
    \\
    &\stackrel{\XX_{p - s, 2}}{\circ}\ar[d]^{}
    \\
    &\stackrel{\XX_{s,1}}{\bullet}
  }\qquad\qquad\qquad
  \xymatrix@=12pt{
    &&\stackrel{\XX_{s,r}}{\bullet}
    \ar@/^/[dl]
    \ar@/_/[dr]
    &\\
    &\stackrel{\XX_{p - s, r - 1}}{\circ}\ar@/^/[dr]
    &
    &\stackrel{\XX_{p - s, r + 1}}{\circ}\ar@/_/[dl]
    \\
    &&\stackrel{\XX_{s,r}}{\bullet}&
  }
\end{equation}
where $r\geq2$.  We also note that the $\PP_{s,r}$ as a $\UresSL2$-module is
isomorphic to the $r$-fold direct sum of $\PP_{s}^\alpha$
with $\alpha=(-1)^{r-1}$. The action of $\E$, $\F$, and $\K$ in $\PP_{s}^\alpha$ can be found in~\cite{[FGST]}.

\subsection{XXZ decompositions}
We are now ready to show the decomposition of $\Hilb_N=\tensor^N\XX_{2,1}$ over
 $\LQG$.  Such a  decomposition for any root of unity and any $N$ was explicitly written in~\cite{GV} with the final result:
\begin{multline}\label{decomp-LQG}
\Hilb_{N}|_{\rule{0pt}{7.5pt}%
\LQG} \cong   \bigoplus_{\substack{s=N\modd2+1,\\s+N=1\modd 2}}^{p-1} \dim\bigl(\IrrTL{\frac{s-1}{2}}\bigr) \XX_{s,1}
\oplus \bigoplus_{r=1}^{r_m-1}\bigoplus_{\substack{s=0,\\rp+s+N=1\modd 2}}^{p-1}\dim\bigl(\IrrTL{\frac{rp+s-1}{2}}\bigr) \PP_{p-s,r}\\
\oplus \bigoplus_{\substack{s=0,\\s+s_m=1\modd 2}}^{s_m+1}\dim\bigl(\IrrTL{\frac{r_mp+s-1}{2}}\bigr) \PP_{p-s,r_m}\,,
\end{multline}
where we defined $N=r_m p+s_m$, for $r_m\in\oN$ and $-1\leq s_m\leq p-2$,
and the multiplicity of each $\PP_{p-s,r}$ and $\XX_{s,1}$ equals
$\dd^0_{\frac{rp+s-1}{2}}$ and $\dd^0_{\frac{s-1}{2}}$,
respectively, which are dimensions of
simple modules $\IrrTL{j}$ over $\TLq{N}$:
\begin{multline}\label{dimIrr}
\dim(\IrrTL{j}) \equiv \dd^0_j = \sum_{n\geq0}\dd_{j+np} -
\sum_{n\geq t(j)+1}\dd_{j+np-1-2(j\modd p)}\\
=\sum_{\substack{j'\geq j,\, (j'-j)\modd p =0,\\-2(j\modd
    p)-1}}  (-1)^{(j'-j)\modd p}\,\dd_{j'},\qquad j\modd p \ne
\ffrac{kp-1}{2},\quad k=0,1,
\end{multline}
where we introduce numbers $ \dd_j = \binom{N}{N/2+j} -
\binom{N}{N/2+j+1}$ and a step function $\stf(j)\equiv\stf$ as
\begin{equation}
\stf=
\begin{cases}\label{stf-def}
1,&\text{for}\quad  j\modd p> \ffrac{p-1}{2},\\
0,&\text{for}\quad  j\modd p< \ffrac{p-1}{2}.
\end{cases}
\end{equation}
We assume that $\dd^0_j=0$ whenever $j>N/2$, as usually.  Each simple module $\IrrTL{j}$ is defined as the head of
the standard $\TLq{N}$-module characterized by $2j$ through lines in the links
representation (see details in e.g.~\cite{M0,ReadSaleur07-2,GV}).


\medskip

 We next turn to an introduction of a lattice
W-algebra which uses the decompositions over $\LQG$ but before we give explicit examples.

\subsubsection{Examples}\label{examp-dec}
 Using~\eqref{decomp-LQG}, we now give several examples of decompositions of the XXZ spin-chains
over $\LQG$ for $p=2,3,4$.
\begin{itemize}
\item For $p=2$, we have
\begin{equation*}
\Hilb_2 = \PP_{1,1},\qquad \Hilb_3 = 2\XX_{2,1} \oplus \XX_{2,2},
\qquad \Hilb_4 = 2\PP_{1,1}\oplus \PP_{1,2},\qquad\dots
\end{equation*}
\item For $p=3$, we have
\begin{gather*}
\Hilb_2 =  \XX_{1,1}\oplus\XX_{3,1},\qquad \Hilb_3 = \XX_{2,1} \oplus \PP_{2,1},
\qquad \Hilb_4 = \XX_{1,1}\oplus 3\XX_{3,1}\oplus \PP_{1,1},\\
\Hilb_5 = \XX_{2,1} \oplus 4\PP_{2,1}\oplus \XX_{3,2}, \qquad \Hilb_6 = \XX_{1,1}\oplus 9\XX_{3,1}\oplus 4\PP_{1,1} \oplus \PP_{2,2},\\
 \Hilb_7 = \XX_{2,1} \oplus 13\PP_{2,1}\oplus 6\XX_{3,2}\oplus \PP_{1,2},\quad
 \Hilb_8 = \XX_{1,1}\oplus 28\XX_{3,1}\oplus 13\PP_{1,1} \oplus 7\PP_{2,2}\oplus\XX_{3,3}, \qquad\dots
\end{gather*}
\item For $p=4$, we have
\begin{gather*}
\Hilb_2 = \XX_{1,1}\oplus\XX_{3,1},\qquad \Hilb_3 = 2\XX_{2,1} \oplus \XX_{4,1},
\qquad \Hilb_4 = 2\XX_{1,1}\oplus2\XX_{3,1} \oplus\PP_{3,1},\\
\qquad\Hilb_5 = 4\XX_{2,1}\oplus4\XX_{4,1}\oplus \PP_{2,1},\qquad
\Hilb_6 = 4\XX_{1,1}\oplus4\XX_{3,1}\oplus 5 \PP_{3,1}  \oplus
\PP_{1,1},\\
\Hilb_{7} =  8\XX_{2,1}\oplus14\XX_{4,1}\oplus 6\PP_{2,1} \oplus \XX_{4,2},\qquad
\Hilb_8 = 8\XX_{1,1}\oplus8\XX_{3,1}\oplus 20 \PP_{3,1}  \oplus
6\PP_{1,1}, \oplus\PP_{3,2}\qquad\dots
\end{gather*}
where the multiplicities in the front of each direct summand are
dimensions of irreducible modules over $\TLq{N}$ -- the centralizer
of $\LQG$.
\end{itemize}
We see from the decompositions in~\bref{examp-dec} a general pattern: the
simple projective module $\XX_{p,2}$ appears for the first time at $N=2p-1$ sites.
This corresponds to the first time when the algebra of endomorphisms
of the module $\Hilb_{N}$ over
 the small quantum group $\UresSL2$ is larger than the
algebra  of endomorphisms respecting the full quantum group $\LQG$ (we discuss how $\LUresSL2$-modules are restricted to
$\UresSL2$ after~\eqref{basis-lusz-irrep-3} and~\eqref{proj-diag-LQG}.) Therefore,  we have isomorphisms
\begin{equation}\label{End-LQG-Ures-smallN}
\Endo_{\LQG}(\chVv)\cong\Endo_{\UresSL2}(\chVv),\qquad \text{for}\;\; 1\leq N\leq 2p-2,
\end{equation}
while we have an inclusion
\begin{equation}
\Endo_{\LQG}(\chVv)\subset\Endo_{\UresSL2}(\chVv),\qquad
\text{for}\;\; N\geq 2p-1.
\end{equation}
We recall that the Temperley--Lieb algebra $\TLq{N}$ is alternatively defined as
$\Endo_{\LQG}(\chVv)$.
\begin{Dfn}\label{Wlat-def}
 \textit{The lattice W-algebra} $\Wlatq{\q}{N}$, for $\q=e^{i\pi/p}$, is the algebra of endomorphisms
  $\Endo_{\UresSL2}(\chVv)$, i.e., the algebra of operators commuting with the action of the restricted quantum
group $\UresSL2$ on the tensor-product representation $\chVv=\tensor_{j=1}^{N}\XX_{2,1}$.
\end{Dfn}
 We note that the algebra $\Wlatq{\q}{2p-1}$ is a non-trivial extension of $\TLq{2p-1}$. It has dimension $\dim \TLq{2p-1} + 3$.
Note that the modules  $\XX_{p,1}$ and $\XX_{p,2}$ restricted to
$\UresSL2$ are $\XX_{p}^+$ and $2\XX^-_{p}$, respectively. This means
that $\Hom_{\UresSL2}(\XX_{p,2},\XX_{p,1})=0$ and
therefore the three additional elements/generators of $\Endo_{\UresSL2}(\Hilb_{2p-1})$ are endomorphisms on the direct summand $\XX_{p,2}$ and
are given by the two maps $\e$ and $\f$
mixing two copies of $\XX_p^-$ in the $s\ell(2)$-doublet $\XX_{p,2}$ and by the generator $\h$ restricted to this direct summand.
The corresponding generators in
 $\Wlatq{\q}{2p-1}$ can be explicitly constructed using the primitive idempotent $\idem_0$
given in App.~\bref{app:idem} --
 the normalized projector ($\idem_0^2=\idem_0$) onto the modules $\XX_{p,r}$ for any \textit{even} $r$ --
as $\W^+_1 = \repQG{2p-1}(\e\idem_0)$, $\W^0_1 = \repQG{2p-1}(\h\idem_0)$, and $\W^-_1 = \repQG{2p-1}(\f\idem_0)$.

We now turn to a definition of the lattice W-algebra $\Wlatq{\q}{N}$
 using similar generators.

\subsection{Generators via the primitive idempotent and $s\ell(2)$}
We consider an extension $\Wlatt{\q}{N}$ of the Temperley--Lieb algebra $\TLq{N}$ by the generators
\begin{gather}\label{wbp-p}
 \W^+_j=\repQG{2p-1}(\e\idem_0)_{j,j+2p-2},\\
\W^0_j=\repQG{2p-1}(\h\idem_0)_{j,j+2p-2},\\
\W^-_j=\repQG{2p-1}(\f\idem_0)_{j,j+2p-2},\label{wbm-p}
\end{gather}
where $j=1,2,\dots,N-2p+2$ and we introduce the notation $\mathsf{x}_{j,j+n}$, with $\mathsf{x}\in\Endo_{\oC}\Hilb_{n+1}$, for the operator $\one\otimes\ldots\otimes\one\otimes\mathsf{x}\otimes\one\otimes\ldots\otimes\one$ applied between the $j$th  and $(j+n)$th sites and acted by identity otherwise.
The algebra $\Wlatt{\q}{N}$ commutes with $\UresSL2$
\begin{equation}\label{Uq-W-comm}
\left[\W^{\alpha}_j,\repQGq(\UresSL2)\right] = 0, \qquad 1\leq j\leq
N-2p+2,\quad \alpha\in\{0,\pm\},
\end{equation}
by construction: $\W^{\alpha}_j$ acts non-trivially on $2p-1$
tensorands between $j$th and $(j+2p-2)$th sites and by the identity on
the other tensorands, the $2p-1$ tensorands consist the module
isomorphic to the module $\Hilb_{2p-1}$ and the $\W^{\alpha}_j$ are
intertwiners respecting the $\UresSL2$-action on this $\Hilb_{2p-1}$. One should  use then the
coassociativity of the comultiplication in $\UresSL2$: we apply $\Delta^{2p-2}(\cdot)$ on the $j$th tensor
component of the element $\Delta^{N-2p+1}(x)$, for any $x\in\UresSL2$, with the resulting
operator $\Delta^{N-1}(x)$ acting on $\chVv$ obviously commuting with $\W^{\alpha}_j$. This finishes our proof of the statement in~\eqref{Uq-W-comm}.

We could write an explicit XXZ expression for $\W^{\alpha}_j$ generators using
the known expression for the central idempotent $\idem_0$ as a polynomial in the Casimir operator, see~\bref{app:idem}:
\begin{equation}
\Delta^n(\idem_0) =
\ffrac{1}{4\prod_{j=1}^{p-1}(2-\q^j-\q^{-j})^2}\bigl(\Delta^n(\casim)+2\bigr)
\prod_{j=1}^{p-1}\bigl(\Delta^n(\casim)-\q^j-\q^{-j}\bigr)^2,
\end{equation}
where
\begin{equation}
\Delta^n(\casim) = (\q-\q^{-1})^2 \Delta^n(\F)\Delta^n(\E)+\q \Delta^n(\K)+\q^{-1}\Delta^n(\K^{-1}).
\end{equation}
Then, the action of~$\repQG{2p-1}(\idem_0)_{j,j+2p-2}$ is obtained using the action~\eqref{XXZrep-EF}, where we appropriately shift the indices of the Pauli $\sigma$ matrices and replace $N$ by $2p-1$.
We also need formulas for $(2p-2)$-folded comultiplication of the divided powers $\e$ and $\f$ applied on the spin-$\half$ chain. They are similarly obtained from~\eqref{XXZrep-e} and~\eqref{XXZrep-f}.
 Using all these formulas one can write the generators $\W^{\alpha}_j$ in terms of Pauli matrices, to obtain expressions similar to TL generators $e_j$, but they  are actually very bulky and we do not give them here.

Note that the operators~\eqref{wbp-p}-\eqref{wbm-p} do not commute with the $U s\ell(2)$ part of $\LQG$ -- they form an $s\ell(2)$ triplet. We also see from the definition which involves the $U s\ell(2)$ generators that for all $j$ the generators $\W^{\alpha}_j$ themself
satisfy $s\ell(2)$ relations.
\begin{prop}\label{prop:sl2-Wlatt}
For $1\leq j\leq N-2p+2$, the generators $\W^+_j$, $\W^0_j$ and
$\W^-_j$ satisfy the $s\ell(2)$ relations
\begin{equation}
 [\W^+_j,\W^-_j] =2\W^0_j,\qquad [\W^0_j,\W^\pm_j]=\pm\W^\pm_j.
\end{equation}
and they span a basis in the adjoint representation of $s\ell(2)$:
\begin{align}
 [\e,\W^+_j]&=0,&&[\e,\W^0_j]=-\W^+_j,&&[\e,\W^-_j]=2\W^0_j,\\
 [\f,\W^+_j]&=-2\W^0_j,&&[\f,\W^0_j]=\W^-_j,&&[\f,\W^-_j]=0.
\end{align}
\end{prop}

We therefore have that there exists an action of $s\ell(2)$ on
$\Wlatt{\q}{N}$ by derivatives of the associative  multiplication.

To get more relations in $\Wlatt{\q}{N}$, we need to express
$\repQG{2p-1}(\idem_0)_{j,j+2p-2}$ in TL generators $e_k$ with $j\leq k\leq
j+2p-3$. We now turn to this problem.

\subsection{Generators via the $\q$-symmetrizer}
We now give an alternative description of the generators using the $\q$-symmetrizers $\qS r$
projecting onto the sub-representation of the maximum spin in the tensor product of spin-$\half$ representations;
we assume now that $\q$ is generic. The intertwiners $\qS r$ are elements from $\TLq{N}$ and defined by the recursion relation
\begin{equation}\label{qSym-def}
\qS{r+1} = \qS{r} -\ffrac{[r]}{[r+1]}\qS{r}e_r\qS{r},
\end{equation}
where $\qS 1=1$ and $e_r$ is the XXZ representation of the TL-generator $e_r$. Using~\eqref{qSym-def}, we easily get first ones
\begin{align*}
&\qS 2 = 1 - \ffrac{1}{[2]}e_1,\\
&\qS 3 = 1 - \ffrac{[2]}{[3]}(e_1+e_2) + \ffrac{1}{[3]}\{e_1,e_2\}.
\end{align*}
We see that at the case of a root of unity the $\qS r$ are degenerate for $r\geq p$ (because $[p]=0$) but not all of them. Indeed, $\qS 2$ is degenerate at $\q=i$ while $\qS 3$ is
well-defined and $\qS 3\to 1-e_1e_2-e_2 e_1$.
In general, the nondegenerate ones are $\qS{np-1}$, for any $n\geq1$.
This follows from the fact that the (projective and irreducible)
module $\XX_{p,n}$ is the limit $\q\to e^{i\pi/p}$ of the
$np$-dimensional Weyl module (generically irreducible module of the spin $\frac{np-1}{2}$) and the $\XX_{p,n}$ is the direct summand
in the limit. Therefore, there exists a projector onto this direct summand and this projector for $N=2p-1$ is simply given by the primitive central idempotent $\idem_0$ of the quantum group.
 We thus observe that
the primitive idempotent $\idem_0$ acted on $\Hilb_{2p-1}$ (and used
above in the definition of the generators $\W^{\alpha}_j$) coincides with $\qS{2p-1}$. For example, at $\q=i$ the operator $\qS 3 = 1-e_1e_2-e_2 e_1$ equals
$\repQG{3}(\idem_0)$ acted on $3$ sites.
This observation allows us to give a more useful definition of the lattice W-algebra generators
\begin{gather}\label{wbp-sym}
 \W^+_j=\repQG{2p-1}(\e)_{j,j+2p-2}\,\qS{2p-1}(j),\\
\W^0_j=\repQG{2p-1}(\h)_{j,j+2p-2}\,\qS{2p-1}(j),\\
\W^-_j=\repQG{2p-1}(\f)_{j,j+2p-2}\,\qS{2p-1}(j),\label{wbm-sym}
\end{gather}
where $\qS{2p-1}(j)$ is the Jones--Wenzl projector of the TL subalgebra generated by $e_k$, with $j\leq k\leq j+2p-2$, {\it i.e.}, it is defined by the recursion relation
\begin{equation}\label{qSym-def-shifted}
\qS{r+1}(j) = \qS{r}(j) -\ffrac{[r]}{[r+1]}\qS{r}(j)e_{r+j-1}\qS{r}(j).
\end{equation}
So, on the full spin-chain, this operator is applied between the $j$th  and $(j+2p-2)$th
sites and acts by identity
on the other tensorands.

We first note that the generators~\eqref{wbp-sym}-\eqref{wbm-sym} of the algebra $\Wlatt{i}{N}$ (at $\q=i$) coincide with the ones introduced in~\bref{sec:lat-W-XX}, where we used a representation of the walled Brauer algebra, up to a sign factor $(-1)^j$.
As in the case of free fermions, the algebra $\Wlatt{\q}{N}$ has many defining relations in addition to the $s\ell(2)$ relations. We do not see any reasons to write  all of them down because they do not actually help in studying representation theory. Moreover, we believe that most of them are relations coming from a particular representation of a bigger algebra which still needs to be discovered. Among the defining relations we can easily get few simple ones. Recall that we have relations in the TL algebra $\qS{r}e_j=e_j\qS{r}=0$, for any $1\leq j\leq r-1$. We therefore have
\begin{equation}
\W^{\alpha}_k e_j = e_j \W^{\alpha}_k = 0, \qquad k\leq j\leq k+2p-3, \quad \alpha = \pm, 0.
\end{equation}

Finally, we believe that $\Wlatt{\q}{N}$ generates the full
centralizer of $\UresSL2$. This is indeed true for $N\leq 2p-1$ and for all $N$ at $p=2$ case where  the definition
\eqref{wbp}-\eqref{wbm} of the lattice W-algebra coincides with $\Wlatt{\q=i}{N}$ defined here
and we have proven Thm.~\bref{thm:centralizer} about the centralizer of $\UresSL2$ for this $\q=i$ case.  We thus finish this subsection with the
following very reasonable conjecture which will be proven elsewhere.
\begin{Conj}
The associative algebra $\Wlatt{\q}{N}$ defined by the generators
$\W^{\alpha}_j$, with $1\leq j\leq N-2p+2$ and $\alpha\in\{0,\pm\}$,
and $e_k$, with $1\leq k\leq N-1$, in
the XXZ representation is isomorphic to the centralizer
$\Wlatq{\q}{N}$ of $\UresSL2$.
\end{Conj}

We next give a short discussion of our results on representation theory of the lattice W-algebra $\Wlatq{\q}{N}$ and relations with  results  on the representation theory of  the chiral  triplet W-algebras in the limit, for $c=-2$ and $c=0$ cases, in particular.

\newcommand{\repsl}[1]{\oC^{#1}}
\subsection{Representation theory of $\Wlatq{\q}{N}$}
By the definition~\bref{Wlat-def}, the algebra $\Wlatq{\q}{N}$ is the full
centralizer of $\UresSL2$,
{\it i.e.}, it is defined as $\Wlatq{\q}{N}=\Endo_{\UresSL2}(\chVv)$ the algebra of all  operators commuting with  the small quantum group $\UresSL2$.
Its representation theory can be thus studied using the decomposition of the XXZ spin-chain over  $\UresSL2$. Such a decomposition can be easily obtained by restriction of
 the $\LQG$-module $\chVv$
in~\eqref{decomp-LQG} on the subalgebra $\UresSL2$. The idea is then to describe the multiplicity spaces in front of $\UresSL2$ direct summands. These multiplicity spaces are then irreducible modules over $\Wlatq{\q}{N}$. The next step is to describe all possible homomorphisms between the $\UresSL2$ direct summands. These homomorphisms represent the action of  $\Wlatq{\q}{N}$ in its reducible and indecomposable modules. We only announce our results about $\Wlatq{\q}{N}$-modules in order to make the paper not too long. All the necessary details will be given in a forthcoming publication~\cite{next}.

\medskip

We begin with the description of projective covers for simple $\Wlatq{\q}{N}$-modules and describe then their content with respect to the TL subalgebra.
Note that $\Wlatq{\q}{N}$ is represented faithfully in the spin
chain, by its definition as the subalgebra $\Endo_{\UresSL2}\chVv$ in the algebra of all
operators acting on~$\chVv$. Therefore, the projective cover of an irreducible $\Wlatq{\q}{N}$-module can be
identified with the indecomposable $\Wlatq{\q}{N}$-submodule in $\chVv$ of maximum dimension and with the property being able to cover the irreducible module.

We have found that the algebra $\Wlatq{\q}{N}$ has the following projective covers.
There are modules $\modwX^\pm_p$ which are irreducible and projective simultaneously and we denote them also by $\modwP^\pm_p$.
There are $p - 1$ cells (taking into account odd and even values of $N$) where reducible but indecomposable modules exist.
Each cell contains projective modules $\modwP^+_s$, $\modwP^-_{p-s}$
and $\modwP^0_s$, with $1\leq s\leq p-1$, and their
subquotient structure is given in Fig.~\ref{fig:W-proj} (from the left to
the right).
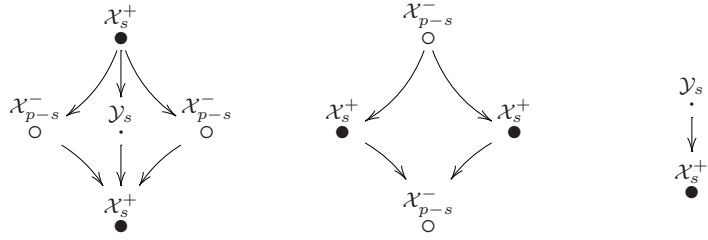
\begin{figure}
\begin{equation*}
   \xymatrix@R=14pt@C=12pt@W=3pt@M=2pt{
    &&\stackrel{\modwX^+_{s}}{\bullet}
    \ar@/^/[dl]
    \ar[d]^{}
    \ar@/_/[dr]
    &\\
    &\stackrel{\modwX^-_{p - s}}{\circ}\ar@/^/[dr]
    &\stackrel{\modwY_{s}}{\cdot}\ar[d]^{}
    &\stackrel{\modwX^-_{p - s}}{\circ}\ar@/_/[dl]
    \\
    &&\stackrel{\modwX^+_{s}}{\bullet}&
  }
\quad
\xymatrix@R=14pt@C=12pt@W=3pt@M=2pt{
    &&\stackrel{\modwX^-_{p-s}}{\circ}
    \ar@/^/[dl]
    \ar@/_/[dr]
    &\\
    &\stackrel{\modwX^+_{s}}{\bullet}\ar@/^/[dr]
    &
    &\stackrel{\modwX^+_{s}}{\bullet}\ar@/_/[dl]
    \\
    &&\stackrel{\modwX^-_{p-s}}{\circ}&
  }
\quad
\xymatrix@R=14pt@C=12pt@W=3pt@M=2pt{
    &&&\\
    &
    &\stackrel{\modwY_{s}}{\cdot}\ar[d]^{}
    &
    \\
    &&\stackrel{\modwX^+_{s}}{\bullet}&
  }
\end{equation*}
\caption{Projective covers of irreducible $\Wlatq{\q}{N}$-modules denoted by $\modwX^+_s$, $\modwX^-_{p-s}$, and  $\modwY_s$. Same diagrams also describe subquotient structure of indecomposable modules over the chiral triplet W-algebra from~\cite{[FGST3]}.}
\label{fig:W-proj}
\end{figure}
These  modules are projective covers of irreducible modules denoted by $\modwX^+_s$, $\modwX^-_{p-s}$, and  $\modwY_s$, respectively.
 The
 decomposition of the irreducible $\Wlatq{\q}{N}$-modules
onto  modules over the two commuting algebras, $\TLq{N}$ and $U s\ell(2)$, can be written as
\begin{align}
 \modwY_s &= \IrrTL{\frac{s-1}{2}}\tensor\repsl{1},&\qquad 1\leq s\leq p-1,\label{modwY-decomp}\\
 \modwX^+_s&=\bigoplus_{r\geq1}\IrrTL{\frac{2rp-s-1}{2}}\tensor\repsl{2r-1},&\qquad
 1\leq s\leq p,\label{modwXp-decomp}\\
 \modwX^-_{p-s}&=\bigoplus_{r\geq1}\IrrTL{\frac{2rp+s-1}{2}}\tensor\repsl{2r},&\qquad
 0\leq s\leq p-1,\label{modwXm-decomp}
\end{align}
where the direct sums are finite,  $\IrrTL{j}$ are irreducible $\TLq{N}$-modules defined after~\eqref{dimIrr} and  as usually we imply that $\IrrTL{j}\equiv0$ whenever $j>N/2$ and  $s+N=1\modd 2$; the tensorands   $\repsl{r}$ denote the
$r$-dimensional irreducible modules over $s\ell(2)$.
 We note that the modules $\modwX^-_{p-s}$ are non-zero only for
$N\geq 2p-1$. For $N<2p-1$, we have modules $\modwY_s$ and
$\modwX^+_{s}$ with trivial $U s\ell(2)$ content and they are
also  irreducible modules over $\TLq{N}$, in accordance
with~\eqref{End-LQG-Ures-smallN}.

The dimensions of irreducible $\Wlatq{\q}{N}$-modules can be expressed
in terms of the multiplicities
 $\dd^0_j$ introduced in~\eqref{dimIrr}.
The dimension of $\modwY_s$ is equal to
$\dd^0_{\frac{s-1}{2}}$ and the module is the singlet
with respect to the $U s\ell(2)$ action.
The dimensions of the irreducible $\Wlatq{\q}{N}$-modules
$\modwX^{\pm}_s$ are
\begin{equation*}
\dim(\modwX^{+}_s) =
\sum_{r\geq1}(2r-1)\dd^0_{\frac{2rp-s-1}{2}},\qquad
\dim(\modwX^{-}_s)= \sum_{r\geq1}2r\,\dd^0_{\frac{(2r+1)p-s-1}{2}}.
\end{equation*}

Finally, we give the $\TLq{N}\tensor U s\ell(2)$
content of the projective $\Wlatq{\q}{N}$-modules
\begin{equation}\label{Wlat-proj-Tl-sl}
\begin{split}
\modwP^+_s|_{\rule{0pt}{7.5pt}%
\TLq{N}\tensor U s\ell(2)} &=
\bigoplus_{n\geq1}\PrTL{\frac{2np-s-1}{2}}\tensor\repsl{2n-1},\qquad
1\leq s\leq p,\\
\modwP^-_{p-s}|_{\rule{0pt}{7.5pt}%
\TLq{N}\tensor U s\ell(2)} &=
\bigoplus_{n\geq1}\PrTL{\frac{2np+s-1}{2}}\tensor\repsl{2n} \qquad
0\leq s\leq p-1,
\end{split}
\qquad s+N=1\modd2,
\end{equation}
where the direct sums are finite of course, similarly to those for the simple modules above, and the subquotient
structure of the projective $\TLq{N}$-modules $\PrTL{j}$ can be found, {\it e.g.}, in~\cite{GV}
and we use the identity
$\PrTL{\frac{rp-1}{2}}=\IrrTL{\frac{rp-1}{2}}$ for modules over $\TLq{N}$.

 We note also that the modules $\modwP^-_{p-s}$ are zero unless
 $N\geq 2p-1$. For $N<2p-1$, we have projective modules $\modwP^0_s$ and
$\modwP^+_{s}$ with trivial $U s\ell(2)$ content -- the two
 subquotients $\modwX^-_{p-s}$ in the first diagram on the left
 in Fig.~\ref{fig:W-proj} are actually absent (because of our
 notation $\IrrTL{j>N/2}=0$ for $\TLq{N}$-modules), and the module $\modwP^+_s$ has
 only three irreducible subquotients. These projectives are
the projective covers over $\TLq{N}$ and the algebra $\Wlatq{\q}{N}$
 has thus the same representation theory whenever $N<2p-1$, in accordance
with~\eqref{End-LQG-Ures-smallN}.

Finally, the  decomposition of the full spin-chain over $\Wlatq{\q}{N}$ can be written as
\begin{equation}\label{decomp-Wlat}
\Hilb_{N}
|_{\rule{0pt}{7.5pt}%
\Wlatq{\q}{N}}
\cong \!\!\!\!\bigoplus_{\substack{s=N\modd2+1,\\s+N=1\modd
    2}}^{p-1}\!\!\!\!\!\!\bigl(
s\modwP^+_s \,\oplus\,
(p-s)\modwP^-_{p-s}\bigr) \; \oplus
\delta_{(p+N)\modd2,1}\,p\,\modwP^+_p\oplus \delta_{N\modd2,1}\,p\,\modwP^-_p,
\end{equation}
where $\delta_{a,b}$ is the usual Kronecker symbol.

\subsection{Two scaling limits: $(p-1,p)$ and $(1,p)$ theories}

In the free fermions case ($p=2$) we already took the scaling limit of our spin-chains, and thus of the lattice W-algebra projective modules. It is easy to see that in this case  the modules $\modwY_s$ are absent and we have only two reducible and indecomposable projectives $\modwP^{\pm}_{1}$ of the ``diamond'' structure (four instead of five subquotients for $\modwP^{+}_{1}$) on even number of sites. For odd number $N$, the spin-chain is semisimple and both the projectives $\modwP^{\pm}_{2}=\modwX^{\pm}_{2}$ are irreducible.   This subquotient structure persists in the limit and give exactly the modules over the triplet W-algebra~\cite{GK}, with the same character of course. This result is not very surprising as long as we consider free theories.

The case $p=2$ enjoys a special symmetry:  the models defined by the two hamiltonians $H(\q)$ and  $-H(\q)$ are equivalent, and exhibit  the same scaling limit. This is not the case for larger values of $p$. While the scaling limit of $H(\q)$ corresponds to the $(p-1,p)$ theory, the hamiltonian $-H(\q)$ corresponds to the $(1,p)$ theory. Physically, this is because two Hamiltonians of opposite sign have in general very different structure of ground state and excited states.

We start by considering the  case of $H(\q)$: for instance, the choice $p=3$ or $\q=e^{\frac{i\pi}{3}}$, which gives rise  to a  $c=0$
logarithmic CFT~\cite{ReadSaleur07-2} in the scaling limit. The generating function of energy levels of the Hamiltonian $H(\q)$ from~\eqref{XXZ-sigma} is well known~\cite{PasquierSaleur} thanks to the Bethe ansatz, and the partition function (or energy levels generating function) for scaling states in the XXZ spin-chain can be written as
\begin{equation}\label{gen-func-lim}
\displaystyle \lim_{N\to\infty} \sum_{{\rm states} \ i} q^{\frac{N}{2\pi v_F} \left(E_i(N) - N e_{\infty} \right)} = q^{-c/24} \sum_{j\geq0}(2j+1)\dfrac{q^{h_{1,1+2j}}-q^{h_{1,-1-2j}}}{\prod_{n=1}^{\infty} \left( 1 - q^n\right)},
\end{equation}
where $v_F=\frac{\pi \sin \gamma}{\gamma}$ (with $2\cos{\gamma}=\q+\q^{-1}$) is the Fermi velocity,  with the central charge
\begin{equation}\label{eqCentralCharge}
\displaystyle c= c_{p-1,p} = 1 - \ffrac{6}{p (p-1)},
\end{equation}
and $E_i(N)$ is the eigenvalue of the $i^{th}$ (counted from the vacuum) eigenstate of $H=-\sum_i e_i$. Here, we have also subtracted from the Hamiltoninan $H$ the density $e_{\infty} = \lim_{N \rightarrow \infty} E_0(N)/N$, with $ E_0(N)$ the ground-state energy; we also use the standard notation for  the conformal weights
\begin{equation}
\displaystyle h_{r,s} = \ffrac{ \left(p r - (p-1)s \right)^2 - 1}{4 p (p-1)}.
\end{equation}
Note that the central charge here is the one from $(p-1,p)$ Minimal Models but the field content is actually different.

Let us denote by  $\Verma_{h}$ the Virasoro Verma module generated from the highest-weight state of weight $h$.
 The expression  at the multiplicity $(2j+1)$ on the right-hand side  of~\eqref{gen-func-lim} coincides
with the Virasoro character ${\rm Tr}\, q^{L_0 - c/24}$ of the so-called \textit{Kac module} with conformal weight $h_{1,1+2j}$ defined
as the quotient
$\Verma_{h_{1,1+2j}}/\Verma_{h_{1,-1-2j}}$ by the submodule corresponding to the singular vector in $\Verma_{h_{1,1+2j}}$ at level $2j+1$.
These Kac modules correspond  to the scaling limit of the standard $\TLq{N}$-modules characterized by $2j$ through lines in the links representation (see details in~\cite{ReadSaleur07-2,GV}).

 It is also well-established~\cite{ReadSaleur07-2} (see also~\cite{GV}) that the simple TL modules $\IrrTL{j}$ go over in the scaling limit to the simple Virasoro  modules with the highest weight $h_{1,2j+1}$, which we will simply denote by $\IrrTL{1,2j+1}$. Then, using the decomposition~\eqref{modwY-decomp} of the  simple $\modwY_s$ modules, we see that they correspond to the  Virasoro irreducible representations from the $(p-1,p)$ minimal model  Kac table. In the case $p=3$ or $c=0$, we have only one module $\modwY_1$ which has dimension $1$ for any $N$ and this module corresponds to the unique operator (identity) in the $(2,3)$ Kac table. For any $p\geq2$, these are also simple modules over the triplet chiral W-algebra first observed in~\cite{[FGST3]}.

Using now the decomposition~\eqref{modwXp-decomp} and~\eqref{modwXm-decomp} of the  simple $\modwX^{\pm}_s$ modules as TL-$Us\ell(2)$ bimoules and the identification of the simple TL modules $\IrrTL{j}$ with the simple Virasoro modules $\IrrTL{1,2j+1}$ in the scaling limit, we  determine the Virasoro  and  $U s\ell(2)$ algebras content for the scaling limit of $\modwX^{\pm}_s$ (we use the same notation for the limits of $\Wlatq{\q}{N}$-modules):
\begin{align}
 \modwX^+_s&=\bigoplus_{r\geq1}\IrrTL{1,2rp-s}\tensor\repsl{2r-1},&\qquad
 1\leq s\leq p,\label{modwXp-limit}\\
 \modwX^-_{p-s}&=\bigoplus_{r\geq1}\IrrTL{1,2rp+s}\tensor\repsl{2r},&\qquad
 0\leq s\leq p-1,\label{modwXm-limit}
\end{align}
where now the direct sums are infinite and are taken over all $s\ell(2)$-spins of a fixed parity; we  also recall that $\IrrTL{1,j}$ is
the simple Virasoro module of the weight $h_{1,j}$. The decompositions~\eqref{modwXp-limit} and~\eqref{modwXm-limit} allow us to
compute the characters ${\rm Tr}\, q^{L_0 - c/24}$ of the limits.
The characters  of  $\modwX^{\pm}_s$  are then expressed exactly as in~\cite{[FGST3]}. Having an identification of
the scaling limit of simples over $\Wlatq{\q}{N}$ with simples over the chiral W-algebra as the Virasoro-$U s\ell(2)$ bimodules,
we believe that the action of our lattice W-algebra on the simples indeed converges in $N\to\infty$ limit to
the action of the chiral W-algebra on the corresponding $\modwX^{\pm}_s$ spaces, as  was shown for $p=2$ case above.

We also note that the limits $\modwY_s$ and $\modwX^{\pm}_s$ we have constructed (in total, $3p-1$ modules)
do not give all the possible irreducible modules (in total, $\half(p-1)(p-2)+2p(p-1)$) for the chiral W-algebra
in the $(p-1,p)$ models constructed in~\cite{[FGST3]}. It is obvious to us that in order to obtain all the
irreducible modules in the  $(p-1,p)$ model we have to use a bigger quantum group discovered in~\cite{[FGST4]} and
the corresponding lattice model, which we leave for a future work.

We can go further and compare also indecomposable but reducible modules. The structure of projective modules $\modwP^{\pm}_{s}$ does not depend on the number of sites and should  persists in the scaling limit. Therefore, we get the same diagrams  in the scaling limit as those in Fig.~\ref{fig:W-proj}. It is interesting to note that indecomposable  modules over the chiral triplet W-algebra for $c=0$ ($p=3$) with the same subquotient structure were proposed in~\cite{GabRunW} (note also that for any $(p-1,p)$ theory similar modules with $5$, $4$, and $2$ subquotients involving the minimal model content were constructed in~\cite{[FGST3]}).

We now consider the Hamiltonian with the opposite sign, that is $-H(\q)$.
Its low-lying spectrum  is  the same
as for the Hamiltonian $H(-\q^{-1})$ which has the Heisenberg coupling with the
opposite sign. The spectra
are related by the similarity transformation
\begin{equation}
-H(\q) = e^{-\rmi\pi\sum_{j=1}^{N}j\sigmaz_j} H(-\q^{-1})e^{\rmi\pi\sum_{j=1}^{N}j\sigmaz_j}.
\end{equation}
This observation
tells us that the generating function of low-lying states for $-H(\q)$ converges to characters from the $(1,p)$ theories. Indeed note that $-\q^{-1} = \q^{p-1}$ and recall that the generating function of low-lying states for $H\bigl(e^{\frac{i\pi(p-p')}{p}}\bigr)$ converges to characters from the $(p',p)$ theory~\cite{PasquierSaleur}. Further, at this limit we observe that the structure of the projective covers $\modwP^{\pm}_{s}$ will be slightly different -- the subquotients $\modwY_s$ corresponding to the minimal model content disappear in the $(1,p)$ scaling limit. It is indeed easy to see for instance that when  $p=3$ when we have just one $\modwY_{s=1}$, and it is spanned by one state which is now {\sl highest}-energy state for the Hamiltonian $-H(\q)$ (or $H(-\q^{-1})$). This state therefore will go to ``infinity'' in the limit to $(1,3)$ or $c=-7$ theory. We expect that similarly all the subquotients   $\modwY_{s}$ have just one point  of condensation of energy levels and they thus disappear in the second, $(1,p)$, scaling limit.  So, all the diagrams of the reducible but indecomposable triplet $W$-algebra projectives will have only $4$ subquotients in the limit corresponding to the $(1,p)$ theories,
as it was expected. The characters, Virasoro-$U s\ell(2)$ bimodule structure on the second scaling limit of simples $\modwX^{\pm}_s$ and the subquotient structure of their projective covers in the limit are indeed the same ones (compare e.g. with~\cite{[FGST2]}) as for the $(1,p)$ triplet W-algebra. Contrary to $(p-1,p)$ models, here we recover all the irreducible modules for the chiral algebra.

\section{Conclusions}

We believe we have reached our main goal of understanding the lattice origin and analog of W-algebra symmetry in chiral (boundary) LCFTs. The next step  should be the analysis of what happens in the non chiral case, that is for bulk LCFTs. In this direction, the results in~\cite{GRS1,GRS2,GRS3} could serve as a good starting point. There is now a lot of evidence that the bulk boundary relationship is extremely complicated in the logarithmic case, and that naive ``periodicized'' versions of the XXZ spin chain or the underlying loop models will not have W-symmetry. What has to be done to restore this symmetry then remains an open question of crucial importance, and we think the study of lattice W-algebras in this context will be the source of further progress.

The case $c=-2$ has also shed light on the mathematical nature of the relationship between the Temperley-Lieb and the Virasoro algebras, or the lattice and continuum W-algebras. For a finite chain, the Fourier modes $H_n^r$ and $W_n^{\alpha,r}$ define a finite dimensional algebra whose structure constants are function of the size $N$. The scaling limit involves not only taking $N\to\infty$ but also normal ordering with respect to the ground state. In this scaling limit, the generators $H_n^r$ and $W_n^{\alpha,r}$ expand on elements from the enveloping algebra of the Virasoro and the triplet $W$-algebra. Infinitely many lattice ``approximations'' converge to the same limit in leading order -- {\it e.g.}, all the $H_n^r$ go to $(-1)^r L_n+L_{-n}$ -- but they differ at next to leading order. While our calculations were considerably simplified by the underlying presence of free fermions, we believe it should be possible to extend the analysis to $\q$ another root of unity. It is important to note that different Hamiltonians will lead, for the same lattice algebra, to different continuum limits, as exemplified above in the $(1,p)$ and $(p-1,p)$ cases.

To finish, we now would like to get back to the
algebra of  zero modes in the $\q=i$ case.
The operators $H^r_0$ and $W^{\alpha,r}_0$ have the following properties
\begin{equation}
 [H^r_0,H^s_0]=0,\qquad [H^r_0,W^{\alpha,s}_0]=0
\end{equation}
and
\begin{align}
 [W^{0,r}_0,W^{+,s}_0]&=-8W^{+,r+s+2}_0+8W^{+,r+s}_0,\\
[W^{0,r}_0,W^{-,s}_0]&=8W^{-,r+s+2}_0-8W^{-,r+s}_0,\\
[W^{+,r}_0,W^{-,s}_0]&=4W^{0,r+s+2}_0-4W^{0,r+s}_0,
\end{align}
for $r,s\in2\oN_0$.
For finite $N$ operators $H^r_0$ and $W^{\alpha,r}_0$ are linearly dependent
\begin{align}
 H^N_0&=\sum_{r=0}^{\frac{N}{2}-1}(-1)^{\frac{N}{2} - r-1}
S_{\frac{N}{2} - r}\bigl(\cos^2\pi\ffrac{0}{N},\cos^2\pi\ffrac{1}{N},\cos^2\pi\ffrac{2}{N},\dots,
        \cos^2\pi\ffrac{N/2-1}{N}\bigr)H^{2r}_0,\\
 W^{\alpha,N-2}_0&=\sum_{r=0}^{\frac{N}{2}-2}(-1)^{\frac{N}{2} - r}
S_{\frac{N}{2} - r-1}\bigl(\cos^2\pi\ffrac{1}{N},\cos^2\pi\ffrac{2}{N},\dots,\cos^2\pi\ffrac{N/2-1}{N}\bigr)W^{\alpha,2r}_0
\end{align}
where $S_r(x_1,x_2,\dots x_k)$ is the $r$th elementary symmetric polynomial in the variables $x_1,x_2,\dots x_k$.
We also recall that $H^0_0$ coincides with the Hamiltonian $H$, thus we have $4N-10$ currents that commute with
the Hamiltonian.

Interestingly, the algebra of $W^{\alpha,r}_0$ can be identified through a quotient algebra of the $s\ell(2)$ loop algebra $s\ell(2)[t,t^{-1}]$ in the following way
\begin{equation}
 W^{+,r}_0=4(t^2-1)t^r e,\qquad W^{-,r}_0=-4(t^2-1)t^r f,\qquad W^{0,r}_0=-8(t^2-1)t^r h,
\end{equation}
where the parameter $t$ satisfies the equation
\begin{equation}\label{zero-mod-rel}
P(t)\equiv(t^2-1) \sum_{r=0}^{\frac{N}{2}-1}(-1)^{r}
S_{r}\bigl(\cos^2\pi\ffrac{1}{N},\cos^2\pi\ffrac{2}{N},\dots,\cos^2\pi\ffrac{N/2-1}{N}\bigr)t^{N-2r-2}=0.
\end{equation}
Recalling the well known property of elementary symmetric polynomials
\begin{equation*}
\prod_{j=1}^n(t-t_j) = \sum_{r=0}^n (-1)^r S_r(t_1,\dots,t_n)t^{n-r},
\end{equation*}
  we obtain that the solutions of the equation~\eqref{zero-mod-rel} are $\{\pm\cos\ffrac{\pi j}{N}, \, 0\leq j\leq N/2-1\}$.
We have thus the representation of the  loop algebra $s\ell(2)\tensor\oC[t,t^{-1}]$ realized by the quotient by the ideal generated by the polynomial $P(t)$ from~\eqref{zero-mod-rel}, {\it i.e.}, we have relations $P(t)e=0$, $P(t)f=0$, {\it etc.}
It is obvious that this representation (or the quotient) of the  $s\ell(2)$ loop algebra is isomorphic to the following tensor-product representation
\begin{equation}\label{tensor-prod-loop}
V[t_1,\dots,t_N]\equiv\bigotimes_{j=0}^{N/2-1}\bigl(\oC^2\bigl[\cos(\pi j/N)\bigr]\otimes\oC^2\bigl[-\cos(\pi j/N)\bigr]\bigr),
\end{equation}
where by $\oC[t_j]$ we denote the so-called evaluation representation of the loop algebra -- any element $p(t) x$ acts by $p(t_j)x$, with $p(t)\in\oC[t,t^{-1}]$ and $x\in s\ell(2)$. The action on the full spin chain~\eqref{tensor-prod-loop} (defined on $N$ sites) is then given by the usual repeated comultiplication of a Lie algebra: $\Delta(x) = x\otimes \one + \one\otimes x$.

Therefore, the zero modes of the $W^{\alpha,r}$ generators, with  $\alpha\in\{\pm,0\}$, can be written as
\begin{equation}\label{W-zero-loop}
W^{\alpha,r}_0 = \sum_{j=1}^N (t_j^2-1)t^r_j\, x_{\alpha}, \qquad x_+=4 e, \; x_0=-8 h,\; x_-=-4f,
\end{equation}
where we also assume that all $t_j$'s are distinct and are from the set $\{\pm\cos\ffrac{\pi j}{N}, \, 0\leq j\leq N/2-1\}$.
Note that this expression is given in a different basis than the usual spin-basis of the XX chains.

We note further that the representation~\eqref{tensor-prod-loop} of the loop algebra is irreducible because all the evaluation parameters $t_j$ are distinct, see {\it e.g.}~\cite{Kac-book}. For the subalgebra generated by the zero modes $W^{\alpha,r}_0$, it is a reducible representation. Indeed the representation~\eqref{tensor-prod-loop} is isomorphic to the initial spin-chain representation where the zero modes are represented by~\eqref{hatWp}-\eqref{hatWm} and the former one is reducible even  for the full lattice W-algebra $\Wlatq{i}{N}$ -- due to presence of the fermionic zero modes and the decomposition onto bosonic and fermionic states.  The spin-chain representation of  $\Wlatq{i}{N}$ is a direct sum of two reducible but indecomposable representations, as discussed earlier in the   more general context of XXZ spin-chains. It is easy to see that the indecomposability here is due the action of non-zero modes $W^{\alpha,r}_n$, as they contain fermionic zero modes $\eta^{\pm}_0$, while the action of $W^{\alpha,r}_0$ is semisimple. The module  $V[t_1,\dots,t_N]$
over the algebra of zero modes is thus semisimple and it is decomposed as
\begin{equation*}
V[t_1,\dots,t_N]|_{\langle W^{\alpha,r}_0\rangle} = 4\bigl(\modwX^+\oplus\modwX^-\bigr),\qquad \text{for even}\; N,
\end{equation*}
where $\modwX^{\pm}$ are irreducible $\Wlatq{i}{N}$-modules and they consist of bosonic and fermionic states, respectively, generated from the vacuum state by negative fermionic modes $\eta^{\pm}_{n<0}$, and $\modwX^{\pm}$ has dimension $2^{N-3}$. Similarly for odd $N$, but we have the multiplicity $2$ instead of $4$.

The algebra of the zero modes $W^{\alpha,r}_0$ is reminiscent of  the Onsager algebra \cite{Ahn} that appears in Onsager's original solution of the Ising model, or the chiral Potts model. Like our algebra of zero modes, the Onsager algebra is also identified as a subalgebra of the $s\ell(2)$ loop algebra \cite{DRoan}, but a different subalgebra from the one we have found here. Now the use of such an identification in the context of the Ising or chiral Potts model is
 that it leads to the spectrum of the Hamiltonian, which in these cases is a particular element of the loop algebra. Things are a bit different for the XX spin chain we are considering. The boundary terms prevent the usual analysis using the Onsager algebra. Meanwhile, we have another ``version'' of this algebra, where the Hamiltonians $H^r_0$ are somehow ``disconnected'', since they commute with all the zero modes. Of course, we  can decide to compute the spectrum of the conserved quantities  $W^{0,r}_0$ nonetheless. Indeed using the identification~\eqref{W-zero-loop} we obtain
\begin{equation}
\mathrm{Spec}(W^{0,r}_0) = \Bigl\{\sum_{j=1}^N m_j (t_j^2-1)t^r_j;\; m_j=\pm4,\, t_j^2=\cos\ffrac{\pi j}{N}, \, 0\leq j\leq \ffrac{N}{2}-1\Bigr\}.
\end{equation}
where $N$ is the number of sites. Of course, this spectrum (the set of eigenvalues of $W^{0,r}_0$) coincides with the eigenvalues extracted  from the fermionic expression~\eqref{hatWz} at $n=0$.

It is natural to wonder at this stage whether higher spin representations, or different choices of the parameters $t_j$, are related with interesting new models exhibiting some sort of W-algebra symmetry. Another question is whether the lattice W-algebras for other roots of unity have anything to do with the $s\ell(2)$ loop algebra.

A last remark of interest concerns the fact that the zero modes $H_0^r$ and $W_0^{\alpha,s}$ commute, which implies that our XX spin chain has more conserved quantities than the $H_0^r$, which are the standard quantities obtainable from derivatives of the transfer matrix with respect to the spectral parameter. Of course, since we are dealing with free fermions, the existence of other conserved quantities apart from the $H_0^r$ is obvious. We have not been able however to extend the commutation of the zero modes to other values of $\q$. It is tempting to speculate that there are other possible lattice models for which such additional conserved quantities could be built, and that they would be better candidates to reproduce the bulk W-symmetric LCFTs in the continuum limit.

\mbox{}\medskip

\noindent {\bf Acknowledgments}: We thank B.L. Feigin,  D. Ridout, V. Schomerus and A.M. Semikhatov  for discussions.
This work was supported in part by the ANR Project 2010 Blanc SIMI 4: DIME. The work of AMG was supported in
part by Marie-Curie IIF fellowship, the RFBR grant 10-01-00408 and the RFBR--CNRS grant 09-01-93105.
The work of IYuT is partially supported
by the grant RFBR 11-02-00685.

\appendix

\section{Quantum groups}\label{app:qunatm-gr-def}
Here, we collect  definitions of different quantum groups at roots
of unity and iterated
comultiplication formulas.  In setting the
notation and recalling the basic facts about $\LUresSL2$ needed below,
we largely follow~\cite{[BFGT]}. We introduce the standard notation $[n]=\ffrac{\q^n-\q^{-n}}{\q-\q^{-1}}$ for $\q$-numbers
 and set $[n]!=[1][2]\dots[n]$.

\subsection{The small or restricted quantum group}
The quantum group $\UresSL2$ is the ``small'' quantum $s\ell(2)$
with $\q = e^{i\pi/p}$ and the  generators $\E$, $\F$, and
$\K^{\pm1}$ satisfying the standard relations for the quantum $s\ell(2)$,
\begin{equation}\label{Uq-relations}
  \K\E\K^{-1}=\q^2\E,\quad
  \K\F\K^{-1}=\q^{-2}\F,\quad
  [\E,\F]=\ffrac{\K-\K^{-1}}{\q-\q^{-1}},
\end{equation}
with  additional relations,
\begin{equation}\label{root-rel}
  \E^{p}=\F^{p}=0,\quad \K^{2p}=\one,
\end{equation}
and the comultiplication is given by
\begin{equation}
  \Delta(\E)=\one\otimes \E+\E\otimes \K,\quad
  \Delta(\F)=\K^{-1}\otimes \F+\F\otimes\one,\quad
  \Delta(\K)=\K\otimes \K.\label{Uq-comult-relations}
\end{equation}
This associative algebra is finite-dimensional, $\mathrm{dim}\UresSL2 = 2p^3$.

\subsection{The full or Lusztig quantum group $\LQG$}
The Lusztig (or full) quantum group $\LQG$
with $\q = e^{i\pi/p}$, for any integer $p\geq2$,  is generated by $\E$, $\F$, and
$\K^{\pm1}$ satisfying the relations~\eqref{Uq-relations} and~\eqref{root-rel},
and
additionally by the divided powers $\f\sim \F^p/[p]!$ and $\e\sim \E^p/[p]!$ which satisfy the usual $s\ell(2)$-relations:
\begin{equation}
  [\h,\e]=\e,\qquad[\h,\f]=-\f,\qquad[\e,\f]=2\h.
\end{equation}
There are also  `mixed' relations
\begin{gather}
  [\h,\K]=0,\qquad[\E,\e]=0,\qquad[\K,\e]=0,\qquad[\F,\f]=0,\qquad[\K,\f]=0,\label{zero-rel}\\
  [\F,\e]= \ffrac{1}{[p-1]!}\K^p\ffrac{\q \K-\q^{-1} \K^{-1}}{\q-\q^{-1}}\E^{p-1},\qquad
  [\E,\f]=\ffrac{(-1)^{p+1}}{[p-1]!} \F^{p-1}\ffrac{\q \K-\q^{-1} \K^{-1}}{\q-\q^{-1}},
    \label{Ef-rel}\\
  [\h,\E]=\ffrac{1}{2}\E A,\quad[\h,\F]=- \ffrac{1}{2}A\F,\label{hE-hF-rel}
\end{gather}
where
\begin{equation}\label{A-element}
  A=\,\sum_{s=1}^{p-1}\ffrac{(u_s(\q^{-s-1})-u_s(\q^{s-1}))\K
        +\q^{s-1}u_s(\q^{s-1})-\q^{-s-1}u_s(\q^{-s-1})}{(\q^{s-1}
         -\q^{-s-1})u_s(\q^{-s-1})u_s(\q^{s-1})}\,
        u_s(\K)\idem_s
\end{equation}
with the polynomials $u_s(\K)=\prod_{n=1,\;n\neq s}^{p-1}(\K-\q^{s-1-2n})$, and
$\idem_s$ are central primitive idempotents defined just below
following~\cite{[FGST]}.
The relations~\eqref{Uq-relations}-\eqref{A-element} are the defining
relations of the quantum group $\LQG$.

\subsection{Central idempotents}\label{app:idem}
We recall the primitive central idempotents in $\UresSL2$~\cite{[FGST]}:
\begin{multline*}
  \idem_s=\ffrac{1}{\psi_s(\beta_s)}\bigl(
  \psi_s(\casim)-\frac{\psi_s^\prime(\beta_s)}{\psi_s(\beta_s)}(\casim
   -\beta_s)\psi_s(\casim)\bigr),\quad
  1\leq s\leq p-1,\\
  \idem_0=\ffrac{1}{\psi_0(\beta_0)}\psi_0(\casim),\qquad
  \idem_p=\ffrac{1}{\psi_p(\beta_p)}\psi_p(\casim),
\end{multline*}
with the polynomials
 \begin{multline*}
   \psi_s(x)=(x-\beta_0)\,(x-\beta_p)
   \smash{\prod_{\substack{j=1\\
         j\neq s}}^{p-1}}(x-\beta_j)^2,\quad 1\leq s\leq p-1,\\
      \psi_0(x)=(x-\beta_p)\prod_{j=1}^{p-1}(x-\beta_j)^2, \quad
      \psi_p(x)=(x-\beta_0)\prod_{j=1}^{p-1}(x-\beta_j)^2,
\end{multline*}
where $\beta_j=\q^j+\q^{-j}$,
and the Casimir element is
\begin{equation*}
  \casim=(\q-\q^{-1})^2 \E\F+\q^{-1}\K+\q \K^{-1}.
\end{equation*}

\subsection{Comultiplication in $\LQG$}
The quantum group $\LQG$ has the Hopf-algebra structure
with the comultiplication
\begin{gather}
  \Delta(\E)=\one\otimes \E + \E\otimes \K,\quad
  \Delta(\F)=\K^{-1}\otimes \F + \F\otimes\one,\quad
  \Delta(\K)=\K\otimes \K,\\
  \Delta(\e)=\e\tensor\one +\K^p\tensor \e
  +\ffrac{1}{[p-1]!} \sum_{r=1}^{p-1}\ffrac{\q^{r(p-r)}}{[r]}\K^p\E^{p-r}\tensor \E^r
  \K^{-r},\label{e-comult}\\
 \Delta(\f)= \f\tensor \one + \K^p\tensor \f+\ffrac{(-1)^p}{[p-1]!}
  \sum_{s=1}^{p-1}\ffrac{\q^{-s(p-s)}}{[s]}\K^{p+s}\F^s\tensor \F^{p-s}.\label{f-comult}
\end{gather}
The antipode and counity are not used in the paper but a reader can
find them, for exmaple, in~\cite{[BFGT]}.

The $N$-folded comultiplication for the capital
generators $\E$ and $\F$ is
\begin{equation}\label{N-fold-comult-cap}
\begin{split}
\Delta^{N}(\E) &=
\sum_{j=1}^{N+1}\underbrace{\one\tensor\dots\tensor\one}_{j-1}\tensor
\E\tensor \K \tensor \dots \tensor \K,\\
\Delta^{N}(\F) &=
\sum_{j=1}^{N+1}\underbrace{\K^{-1}\tensor\dots\tensor \K^{-1}}_{j-1}\tensor
\F\tensor \one \tensor \dots \tensor \one,
\end{split}
\end{equation}
and for the divided powers $\e$ and $\f$ is
\begin{align}
 \Delta^N(\e)&=\sum_{n=1}^{N+1}  \bigotimes_{j=1}^{n-1} \K^p\tensor \e\tensor\bigotimes_{j=n+1}^{N+1} \one&\label{N-fold-comult-e}\\
& \mbox{}\qquad\qquad +\sum_{\stackrel{0\leq j_1,j_2,\dots,j_{N+1}\leq p-1}{j_1+j_2+\dots+j_{N+1}=p}}\frac{\q^{\sum_{r<s}j_rj_s}}{[j_1]![j_2]!\dots[j_{N+1}]!}
  \bigotimes_{m=1}^{N+1}
  \E^{j_m}\K^{p+\sum_{s=1}^{m-1}j_s},&\notag\\
 \Delta^N(\f)&=\sum_{n=1}^{N+1}  \bigotimes_{j=1}^{n-1} \K^p\tensor \f\tensor\bigotimes_{j=n+1}^{N+1} \one&\label{N-fold-comult-f}\\
 & \mbox{}\qquad\quad+(-1)^{p-1}\!\!\sum_{\stackrel{0\leq j_1,j_2,\dots,j_{N+1}\leq p-1}{j_1+j_2+\dots+j_{N+1}=p}}\frac{\q^{\sum_{r<s}j_rj_s}}{[j_1]![j_2]!\dots[j_{N+1}]!}
  \bigotimes_{m=1}^{N+1}
  \F^{j_m}\K^{p+\sum_{s=1}^{m}j_s}.&\notag
\end{align}

\subsection{Standard spin-chain notations}
We note the Hopf-algebra homomorphism
\begin{equation*}
\E\mapsto S^+ k, \qquad \F\mapsto k^{-1} S^-,\qquad \text{with}\quad k=\sqrt{\K},
\end{equation*}
where we introduced the more usual (in the spin-chain
literature~\cite{PasquierSaleur,DFMC}\footnote{We note that our
  convention for the spin-chain  representation differs from the one
  in~\cite{PasquierSaleur} by the change $\q\to\q^{-1}$.}) quantum group generators
 \begin{equation}
 S^{\pm}=\sum_{1\leq j\leq N} \q^{-\sigma_1^z/2}\otimes\ldots
 \q^{-\sigma_{j-1}^z/2}\otimes\sigma_{j}^{\pm}\otimes\q^{\sigma_{j+1}^z/2}\otimes
 \ldots \otimes \q^{\sigma_{N}^z/2}
 \end{equation}
together with $k=\q^{S^z}$ and the relations
\begin{gather*}
    k S^{\pm}k^{-1}=\q^{\pm 1}S^{\pm},\qquad
   \left[S^{+},S^{-}\right]=\ffrac{k^2-k^{-2}}{\q-\q^{-1}},\\
    \Delta(S^{\pm})=k^{-1}\otimes S^{\pm}+S^{\pm}\otimes k.
\end{gather*}

\subsubsection{The case of XX spin-chains}
For $p=2$ or ``XX spin-chain'' case, the  $(N-1)$-folded coproduct of the renormalized powers $\e$ and $\f$ reads
\begin{multline}\label{N-fold-comult-ren-e}
\Delta^{N-1}\e =
\sum_{j=1}^{N}\underbrace{\one\tensor\dots\tensor\one}_{j-1}\tensor
\e\tensor \K^2 \tensor \dots \tensor \K^2 +\\
+ \q\sum_{t=0}^{N-2}\sum_{j=1}^{N-1-t} \underbrace{\one\tensor\dots\tensor\one}_{j-1}\tensor
\E\tensor \underbrace{\K \tensor \dots \tensor \K}_{t}\tensor \E\K\tensor \K^2 \tensor \dots \tensor \K^2
\end{multline}
and
\begin{multline}\label{N-fold-comult-ren-f}
\Delta^{N-1}\f =
\sum_{j=1}^{N}\underbrace{\K^2\tensor\dots\tensor \K^2}_{j-1}\tensor
\f\tensor \one \tensor \dots \tensor \one +\\
+ \q^{-1}\sum_{t=0}^{N-2}\sum_{j=1}^{N-1-t}
\underbrace{\K^2\tensor\dots\tensor \K^2}_{j=1}\tensor
\K^{-1}\F\tensor \underbrace{\K^{-1} \tensor \dots \tensor \K^{-1}}_{t}\tensor \F\tensor \one \tensor \dots \tensor \one.
\end{multline}
These renormalized powers can also be expressed in terms of the more usual spin-chain operators, and one finds at $p=2$
\begin{equation*}
\Delta^{N-1}(\e)=\q S^{+(2)}k^{2},\qquad \Delta^{N-1}(\f)=\q^{-1}
k^{-2}S^{-(2)},
\end{equation*}
where $\q=i$ and
 \begin{equation}
 S^{\pm (2)}=\sum_{1\leq j<k\leq N-1} \q^{-\sigma_1^z}\otimes\ldots
 \otimes \q^{-\sigma_{j-1}^z}\otimes\sigma_{j}^{\pm}
 \otimes1\otimes\ldots\otimes 1\otimes \sigma_k^{\pm} \otimes
 \q^{\sigma_{k+1}^z}\otimes \ldots \otimes \q^{\sigma_{N}^z}.
 \end{equation}

\section{The proof about the centralizer of $\UresSL2$ for $\q=i$}\label{app:Ures-centr}
In this section we prove the following theorem.
\begin{Thm}\label{thm:centralizer}
The centralizer of the $\UresSL2$ representation at $\q=i$ -- the open XX spin-chain -- is generated by
\begin{equation*}
e_k=c_k c_{k+1}^\dagger + c_{k+1} c_k^\dagger
  +i\bigl(c_k^\dagger c_k-c^\dagger_{k+1}c_{k+1}\bigr)
\end{equation*}
and
\begin{gather*}
 \W^+_j=(-1)^{j}\bigl(c^\dagger_jc^\dagger_{j+1}+i c^\dagger_jc^\dagger_{j+2}
-c^\dagger_{j+1}c^\dagger_{j+2}\bigr),\\
\W^0_j=-\ffrac{1}{2}(1 + i c^\dagger_j c_{j+1} - c^\dagger_j c_{j+2} + i c^\dagger_{j+1} c_j -
   2 c^\dagger_{j+1} c_{j+1} - i c^\dagger_{j+1} c_{j+2} - c^\dagger_{j+2} c_j -
  i c^\dagger_{j+2}c_{j+1}),\\
\W^-_j=(-1)^{j+1}\bigl(c_jc_{j+1}+i c_jc_{j+2}
-c_{j+1}c_{j+2}\bigr),
\end{gather*}
where  $1\leq k\leq N-1$ and $1\leq j\leq N-2$.
\end{Thm}
Our proof is computational and consists of few lemmas. Recall first the linear combinations~\eqref{theta-ferm} and~\eqref{etapm} of the $c_j$ and $c_j^{\dagger}$ operators -- the Fourier transforms  $\eta^{\pm}_n$ of the lattice fermions. We recall also the definition~\eqref{etapm-zero} of the zero modes $\eta^{\pm}_0$ and conjugate to them $\gamma^{\mp}$ operators defined in~\eqref{gammap}-\eqref{gammam}. The $\eta^{\pm}_n$, with  $-\frac{N}{2}+1\leq n\leq \frac{N}{2}-1$, and $\gamma^{\pm}$ generate a Clifford algebra of dimension $2^{2N}$, see relations in~\eqref{eta-comm}. Note that the spin-chain is an irreducible representation of this Clifford algebra. We also recall that each pair of the modes $\eta^{\pm}_n$, and the pair $\gamma^{\pm}$, form an $s\ell(2)$ doublet, with respect to the  $s\ell(2)$ representation given in~\eqref{sl-def-1}-\eqref{sl-def-2}, see its action  in~\eqref{sl-act-1} and~\eqref{sl-act-2}. Therefore, the $\eta^{\pm}_n$ and $\gamma^{\pm}$ change an eigenvalue of the Cartan $S^z=2\h$ by~$\pm1$. Now, we can formulate and prove the following technical lemma.

\begin{lemma}\label{lem:centr-gl}
The centralizer of the $\gl(1|1)$, {\it i.e.}, the walled Brauer algebra representation in the XX spin-chain is generated by the fermionic bilinears $\eta^{+}_n\eta^-_m$, with $-\frac{N}{2}+1\leq n,m\leq \frac{N}{2}-1$ and for any $N\in\oN$.
\end{lemma}
\begin{proof}
Define an index set $I$ as the set of integers $-\frac{N}{2}+1\leq n\leq \frac{N}{2}-1$.
We first note that the centralizer of the algebra generated by $S^z$ -- a commutative subalgebra in the $\gl(1|1)$ representation -- has the fermionic bilinears
\begin{equation}\label{ferm-bil-1}
\eta^{+}_n\eta^-_m, \qquad \text{with}\;  n,m\in I,
\end{equation}
 and
 \begin{equation}\label{ferm-bil-2}
 \gamma^+\gamma^-, \quad \gamma^+\eta^-_n, \quad \eta^+_n\gamma^-, \qquad \text{with}\; n\in I,
 \end{equation}
  as its generators. Indeed, since  the $\eta^{\pm}_n$ and $\gamma^{\pm}$ -- the generators of the Clifford algebra -- change an eigenvalue of  $S^z$ by~$\pm1$ it is clear that only monomials in  $\eta^{\pm}_n$ and $\gamma^{\pm}$ with equal number of `$+$' and `$-$' commute with $S^z$. All such monomials can obviously be generated by the bilinears in~\eqref{ferm-bil-1} and~\eqref{ferm-bil-2}, with the use of the (anti)commutation relations in the Clifford algebra.

  We then check by direct calculations that any linear combination of the bilinears in~\eqref{ferm-bil-2} does not commute with the fermionic part of $\gl(1|1)$ -- the zero modes $\eta^{\pm}_0$ -- while any linear combination of $\eta^{+}_n\eta^-_m$ does. It is also straightforward to see that any word from the Clifford algebra involving the bilinears from~\eqref{ferm-bil-2} and any linear combination of such words do not commute with the $\eta^{\pm}_0$. Finally, recall that the $\gl(1|1)$ representation in our spin-chain is generated by $S^z$ and $\eta^{\pm}_0$. We therefore conclude that the centralizer of the  $\gl(1|1)$ action is generated by the $\eta^{+}_n\eta^-_m$, with $n,m\in I$.
\end{proof}

Consider then a subalgebra $\repQGgl(\UresSL2)$ in the (image of) $\gl(1|1)$, see the definition of $\repQGgl$ in~\eqref{QG-fermf-1}. This subalgebra is generated by $\K=(-1)^{S^z}$ and  $\eta^{\pm}_0$. Our aim is to describe the centralizer of the $\UresSL2$ action. Obviously the centralizer of $\K$ is generated by all bilinears in the fermionic modes $\eta^{\pm}_n$ and $\gamma^{\pm}$. Following the same lines as in the proof of Lem.~\bref{lem:centr-gl} we  prove our second technical lemma.
\begin{lemma}\label{lem:centr-Uq}
The centralizer of the $\UresSL2$ for $\q=i$ is generated by the fermionic bilinears $\eta^{+}_n\eta^-_m$, $\eta^{+}_n\eta^+_m$, and $\eta^{-}_n\eta^-_m$, with $-\frac{N}{2}+1\leq n,m\leq \frac{N}{2}-1$ and for any $N\in\oN$.
\end{lemma}

Then, note that the centralizer of $\UresSL2$ is in the image under the adjoint action of the $s\ell(2)$ algebra on the centralizer of $\gl(1|1)$, {\it i.e.}, on the walled Brauer algebra representation. Indeed, in terms of  generators for both the centralizers described in Lem.~\bref{lem:centr-gl} and Lem.~\bref{lem:centr-Uq} we have
\begin{equation}\label{sl2-bilin}
[\e,\eta^{+}_n\eta^-_m] = \eta^{+}_n\eta^+_m,\qquad [\f,\eta^{+}_n\eta^-_m] = \eta^{-}_n\eta^-_m,\qquad\quad -\ffrac{N}{2}+1\leq n,m\leq \ffrac{N}{2}-1,
\end{equation}
where we use the fermionic  expressions~\eqref{sl-def-1} of the $\e$ and $\f$ generators of the $s\ell(2)$ algebra. We also note that the action of $s\ell(2)$ differentiates the multiplication in the Clifford algebra, {\it i.e.}, for $a$ and $b$ from the Clifford algebra, we have $\e(ab)=\e(a)b + a\e(b)$, {\it etc.} Therefore, the equations~\eqref{sl2-bilin} define the image of any element of the walled Brauer algebra representation under the $s\ell(2)$ action. The centralizer of $\UresSL2$ is thus generated by (actually, coincides with) this image  of the walled Brauer algebra under the action of $U s\ell(2)$. We finally note that by definition~\eqref{sl2act-WB}-\eqref{wbm}  the lattice W-algebra  $\Wlatq{i}{N}$ is also generated by the image under the action of $U s\ell(2)$ on the same walled Brauer algebra where just another system of generators, the $e_j$ and $\W_j$ instead of $\eta^{+}_n\eta^-_m$, is used.  The $\Wlatq{i}{N}$ is thus isomorphic to the centralizer of $\UresSL2$ which finishes the proof of Thm.~\bref{thm:centralizer}.
%

\section{Examples of relations in the lattice W-algebra\label{App:W-rel}}
Here, we give some examples of relations in the lattice W-algebra $\Wlatq{i}{N}$.
\begin{itemize}
 \item[\bf $\bullet$\; $3$ sites:] the defining relations are
 \begin{gather*}
  \W^\pm_{m}\W^\pm_{m}=0,\\
 \W^0_m\W^\pm_m=\pm\ffrac{1}{2}\W^\pm_m,\quad \W^0_{m}\W^0_{m}=\ffrac{1}{4}(1-e_me_{m+1}-e_{m+1}e_m),\\
\W^+_m\W^-_m=\W^0_m+\ffrac{1}{2}(1-e_me_{m+1}-e_{m+1}e_{m}),\\
e_m\W^\alpha_{m}=\W^\alpha_{m}e_{m}=e_{m+1}\W^\alpha_{m}=\W^\alpha_{m}e_{m+1}=0.
\end{gather*}
These relations allow to construct a basis in the lattice W-algebra on $3$ sites. The elements
\begin{equation*}
 \W^\alpha_1,\qquad \alpha=\pm,0,
\end{equation*}
together with  $5$ basis elements from the TL algebra give $8$ basis elements in total.
\item[\bf $\bullet$\; $4$ sites:]  let us introduce notations
$e_{1,m}=[e_m,e_{m+1}]$, $e_{2,m}=\bigl[e_m,[e_{m+1},e_{m+2}]\bigr]$, etc., and $\W^{\alpha}_{1,m}=[\W^{\alpha}_{m},e_{m+1}]$. Then,
   the defining relations in $\Wlatq{i}{4}$ can be written as (we set $m=1$ here)
\begin{gather*}
 \W^\pm_{m}\W^\pm_{m+1}=0,\\
 \W^0_{m}\W^0_{m+1}=\ffrac{1}{4}(1-2e_{2,m}e_{m+1}+2e_{1,m}),\\
\W^0_{m+1}\W^0_{m}=\ffrac{1}{4}(1-2e_{2,m}e_{m+1}-2e_{1,m+1}),\\
\W^+_m\W^-_{m+1}=\W^0_{m+1}-e_{m+1}e_m\W^0_{m+1}+e_{1,m+1}-e_{m+1}e_{2,m}+\ffrac{1}{2},\\
\W^+_{m+1}\W^-_{m}=\W^0_{m}+e_{m+1}e_m\W^0_{m+1}-e_{1,m+1}-e_{m+1}e_{2,m}+\ffrac{1}{2}
\end{gather*}
in particular
\begin{align*}
 [\W^+_{m+1},\W^-_m]&=\W^0_m+\W^0_{m+1}-e_{1,m}-e_{1,m+1},\\
[\W^+_{m},\W^-_{m+1}]&=\W^0_m+\W^0_{m+1}+e_{1,m}+e_{1,m+1},
\end{align*}
and relations with the TL generators
\begin{equation*}
\begin{array}{c}
e_{m-1}\W^\alpha_{m}e_{m-1}=e_{m+2}\W^\alpha_{m}e_{m+2}=0,\\[8pt]
  {}[e_{m},\W^\alpha_{m+1}]=[\W^\alpha_{m},e_{m+2}],\\[8pt]
   e_{1,m}\W^\alpha_{m+1}= e_{1,m+1}\W^\alpha_{m},\qquad
  \W^\alpha_{m+1}e_{1,m}=\W^\alpha_{m}e_{1,m+1},\\[8pt]
   e_{2,m}\W^\alpha_{m}= -\half e_{m}\W^\alpha_{m+1},\qquad
  \W^\alpha_{m}e_{2,m}=-\half\W^\alpha_{m+1}e_{m},\\[8pt]
[e_{m+1},\W^\alpha_{1,m}]=-\half\W^\alpha_{m}+\half\W^\alpha_{m+1},\\[8pt]
\end{array}
\kern+5pt\alpha=+,-,0.
\end{equation*}
and
\begin{gather*}
  \W^\pm_{m}e_{m+2}\W^\pm_{m}=\W^\pm_{m+1}e_m\W^\pm_{m+1}=0,\\
 \W^0_{m}e_{m+2}\W^+_{m}=\W^0_{m+1}e_m\W^+_{m+1}=\ffrac{1}{2}(\W^+_{m}e_{m+2}+\W^+_{m+1}e_m),\\
 \W^+_{m}e_{m+2}\W^0_{m}=\W^+_{m+1}e_m\W^0_{m+1}=-\ffrac{1}{2}(e_{m+2}\W^+_{m}+e_m\W^+_{m+1}),\\
\W^0_{m}e_{m+2}\W^0_{m}=\W^0_{m+1}e_m\W^0_{m+1}=\ffrac{1}{4}(e_m+e_{m+2}-2e_{2,m}\kern100pt\\
              \kern170pt+2e_{1,m}e_{m+2}+2e_me_{1,m+1}-2e_me_{m+1}e_{m+2}),\notag\\
 \W^+_{m}e_{m+2}\W^-_{m}=\W^+_{m+1}e_m\W^-_{m+1}=e_{m+2}\W^0_{m}+e_m\W^0_{m+1}\kern80pt\\
  \kern100pt+e_me_{m+1}e_{m+2}
    -\ffrac{1}{2}(e_m+e_{m+2}-2e_{2,m}
              +2e_{1,m}e_{m+2}+2e_me_{1,m+1})\notag\\
 \W^-_{m}e_{m+2}\W^+_{m}=\W^-_{m+1}e_m\W^+_{m+1}=-\W^0_{m}e_{m+2}-\W^0_{m+1}e_m\kern80pt\\
  \kern100pt-e_me_{m+1}e_{m+2}
    +\ffrac{1}{2}(e_m+e_{m+2}-2e_{2,m}
              +2e_{1,m}e_{m+2}+2e_me_{1,m+1}).\notag
\end{gather*}
These relations allow to construct a basis in the lattice W-algebra $\Wlatq{i}{4}$:
\begin{gather*}
 \W^\alpha_m,\quad\W^\alpha_{m+1},\qquad\mbox{$6$ elements ,}\\
 e_{m+2}\W^\alpha_m,\quad\W^\alpha_me_{m+2},\quad e_m\W^\alpha_{m+1},\qquad\mbox{$9$ elements,}\\
 e_{m+1}e_m\W^\alpha_{m+1},\qquad\mbox{$3$ elements},
\end{gather*}
which together with $14$ basis elements from the TL algebra give $32$ basis elements.

\item[\bf $\bullet$\; $5$ sites:] for $5$ sites the relations are very bulky and we give only simple examples of them:
\begin{align*}
 [\W^+_{m+2},\W^-_m]&=[e_m,[e_{m+1},\W^0_{m+2}]]-e_{3,m},\\
[\W^+_{m},\W^-_{m+2}]&=[e_m,[e_{m+1},\W^0_{m+2}]]+e_{3,m},\\
[\W^0_{m},\W^0_{m+2}]&=2e_{3,m}.
\end{align*}

\end{itemize}

\section{The triplet W-agebra currents in symplectic fermions\label{App:W-tripl}}
We recall that the  triplet $W$-algebra currents in terms of the symplectic fermions \cite{Kausch}
\begin{equation*}
 \eta^\pm(z)=\sum z^{-n-1}\eta^\pm_n,\qquad \text{with} \quad \eta^+(z)\eta^-(w)=\frac{1}{(z-w)^2}+\dots,
\end{equation*}
where the sum is assumed to be over $n\in\oZ$ an the dots stand for regular terms, are given as
\begin{align}
 T(z)&=-:\eta^+(z)\eta^-(z):=\sum z^{-n-2}L_n,\label{eq:T}\\
W^{\pm}(z)&=:\partial\eta^{\pm}(z)\eta^{\pm}(z):=\sum z^{-n-3}W^{\pm}_n,\\
W^0(z)&=\half\bigl(:\partial\eta^+(z)\eta^-(z):-:\eta^+(z)\partial\eta^-(z):\bigr)=\sum z^{-n-3}W^0_n.
\end{align}
We also introduce modes of some composite currents
\begin{equation*}
 :T(z)T(z):=\sum z^{-n-4}(L^2)_n,\qquad
:T(z)W^\alpha(z):=\sum z^{-n-5}(LW^\alpha)_n.
\end{equation*}
These modes are expressed in the modes of the symplectic fermions $\eta^{\pm}(z)$ as
\begin{align}
L_n&=-\sum_{j_1+j_2=n}:\eta^+_{j_1}\eta^-_{j_2}:,\label{L-modes}\\
(L^2)_n&=-\half\sum_{j_1+j_2=n}((j_1+1)(j_1+2)+(j_2+1)(j_2+2)):\eta^+_{j_1}\eta^-_{j_2}:,\label{L2-modes} \\
W^{\pm}_n&=-\sum_{j_1+j_2=n}(j_1+1)\eta^{\pm}_{j_1}\eta^{\pm}_{j_2},\label{Wpm-modes}\\
W^0_n&=-\sum_{j_1+j_2=n}(j_1-j_2)\eta^+_{j_1}\eta^-_{j_2},\label{W0-modes}\\
(LW^+)_n&=\sum_{j_1+j_2=n}(-\ffrac{1}{3} (j_2 + 1) (j_2 + 2) (j_2 + 3) + \half (j_1 + 1) (j_2 + 1) (j_2 + 2))\eta^+_{j_1}\eta^+_{j_2}.\label{LW-modes}
\end{align}
The modes of these currents satisfy the commutation relations \cite{GK}
\begin{align}
\label{LL}[L_n,L_m]&=(n-m)L_{n+m}-\frac{1}{6}n(n^2-1)\delta_{n+m,0},\\
\label{LW}[L_n,W^\alpha_m]&=(2n-m)W^\alpha_{n+m},\\
 \label{WW}[W^\alpha_n,W^\beta_m]&=g^{\alpha\beta}\Bigl( 2(n-m)(L^2)_{n+m}-\frac{1}{4} (n-m)(2 n+m+4)( 2 m + n+4) L_{n+m}-\\
\notag&\kern210pt
-\frac{1}{120} n (n^2 - 1) (n^2 - 4)\delta_{n+m,0}\Bigr)+\\
\notag&+f^{\alpha\beta}_\gamma\Bigl(\frac{12}{5}(LW^\gamma)_{n+m}+\frac{1}{10}(2n^2+2m^2-21nm-36n-36m-76)W^\gamma_{n+m}\Bigr),
\end{align}
where $g^{\alpha\beta}$ and $f^{\alpha\beta}_\gamma$ are defined after~\eqref{w-ope}.

Our aim in the paper is to show that these commutation relations can be
obtained from the scaling limit of those in the (finite-dimensional) lattice $W$-algebra $\Wlatq{i}{N}$ introduced in Sec.~\bref{sec:lat-W-XX} and~\bref{sec:gen-hamilt} as an extension of the Temperley--Lieb algebra.

\section{Commutation relations in the $\Wlatq{i}{N}$
  algebra\label{sec:comm-rel}}
We collect here a  list of commutators of the Fourier modes in  $\Wlatq{i}{N}$ generators:
\begin{align}\label{H0W0}
 [H^0_n,W^{\alpha,0}_m]&=2\sin\pi\ffrac{2n-m}{2N}W^{\alpha,1}_{n+m}+2\sin\pi\ffrac{2n+m}{2N}W^{\alpha,1}_{n-m},\\
[H^0_n,W^{\alpha,1}_m]&=2\sin\pi\ffrac{3n-m}{2N}W^{\alpha,2}_{n+m}-2\sin\pi\ffrac{3n+m}{2N}W^{\alpha,2}_{n-m}-\\
&\notag\kern50pt-\sin\pi\ffrac{n}{N}\bigl(\cos\pi\ffrac{n-m}{2N}W^{\alpha,0}_{n+m}
    -\cos\pi\ffrac{n+m}{2N}W^{\alpha,0}_{n-m}\bigr),\\
[H^1_n,W^{\alpha,0}_m]&=2\sin\pi\ffrac{n-m}{N}W^{\alpha,2}_{n+m}+2\sin\pi\ffrac{n+m}{N}W^{\alpha,2}_{n-m}+\\
&\notag\kern50pt+2\cos\pi\ffrac{n}{2N}\sin\pi\ffrac{m}{2N}\bigl(\cos\pi\ffrac{n-m}{2N}W^{\alpha,0}_{n+m}
    -\cos\pi\ffrac{n+m}{2N}W^{\alpha,0}_{n-m}\bigr),\\
\label{H1W1}[H^1_n,W^{\alpha,1}_m]&=2\sin\pi\ffrac{3n-2m}{2N}W^{\alpha,3}_{n+m}-2\sin\pi\ffrac{3n+2m}{2N}W^{\alpha,3}_{n-m}+\\
&\notag\kern50pt+\bigl(\cos\pi\ffrac{n}{N}\sin\pi\ffrac{n}{2N}-\sin\pi\ffrac{3n-2m}{2N}\bigr)W^{\alpha,1}_{n+m}-\\
    &\notag\kern140pt-\bigl(\cos\pi\ffrac{n}{N}\sin\pi\ffrac{n}{2N}-\sin\pi\ffrac{3n+2m}{2N}\bigr)W^{\alpha,1}_{n-m},
\end{align}
\begin{align}\label{W0W0}
 [W^{\alpha,0}_n,W^{\beta,0}_m]&=f^{\alpha\beta}_\gamma
\Bigl(2\cos\pi\ffrac{n-m}{N}W^{\gamma,2}_{n+m}+2\cos\pi\ffrac{n+m}{N}W^{\gamma,2}_{n-m}-\\
&\notag\kern50pt -\half \bigl(\cos\pi\ffrac{n}{N}+\cos\pi\ffrac{m}{N}+\cos\pi\ffrac{n-m}{N}+1\bigr)W^{\gamma,0}_{n+m}-\\
&\notag\kern120pt -\half \bigl(\cos\pi\ffrac{n}{N}+\cos\pi\ffrac{m}{N}+\cos\pi\ffrac{n+m}{N}+1\bigr)W^{\gamma,0}_{n-m}\Bigr)+\\
&\notag+g^{\alpha\beta}\Bigl(4\sin\pi\ffrac{n-m}{N}H^3_{n+m}+4\sin\pi\ffrac{n+m}{N}H^3_{n-m}-\\
&\notag\kern50pt-\bigl(\sin\pi\ffrac{n}{N}-\sin\pi\ffrac{m}{N}+3\sin\pi\ffrac{n-m}{N}\bigr)H^1_{n+m}-\\
&\notag\kern90pt-\bigl(\sin\pi\ffrac{n}{N}+\sin\pi\ffrac{m}{N}+3\sin\pi\ffrac{n-m}{N}\bigr)H^1_{n-m}\Bigr),\\
 [W^{\alpha,0}_n,W^{\beta,1}_m]&=f^{\alpha\beta}_\gamma
\Bigl(2\cos\pi\ffrac{3n-2m}{2N}W^{\gamma,3}_{n+m}-2\cos\pi\ffrac{3n+2m}{2N}W^{\gamma,3}_{n-m}-\\
&\notag\kern50pt -\bigl(\cos\pi\ffrac{n}{2N}\cos\pi\ffrac{n}{N}+\cos\pi\ffrac{3n-2m}{2N}\bigr)W^{\gamma,1}_{n+m}+\\
&\notag\kern120pt + \bigl(\cos\pi\ffrac{n}{2N}\cos\pi\ffrac{n}{N}+\cos\pi\ffrac{3n+2m}{2N}\bigr)W^{\gamma,1}_{n-m}\Bigr)+\\
&\notag+g^{\alpha\beta}\Bigl(4\sin\pi\ffrac{3n-2m}{2N}H^4_{n+m}-4\sin\pi\ffrac{3n+2m}{2N}H^4_{n-m}-\\
&\notag\kern50pt-2\bigl(\cos\pi\ffrac{n}{2N}\sin\pi\ffrac{n}{N}+2\sin\pi\ffrac{3n-2m}{2N}\bigr)H^2_{n+m}+\\
&\notag\kern90pt+2\bigl(\cos\pi\ffrac{n}{2N}\sin\pi\ffrac{n}{N}+2\sin\pi\ffrac{3n+2m}{2N}\bigr)H^2_{n-m}+\\
&\notag\kern110pt+2\sin\pi\ffrac{n}{N}\cos\pi\ffrac{m}{2N}\cos\pi\ffrac{n-m}{2N}H^0_{n+m}-\\
&\notag\kern145pt-2\sin\pi\ffrac{n}{N}\cos\pi\ffrac{m}{2N}\cos\pi\ffrac{n+m}{2N}H^0_{n-m}\Bigr),\\
\label{W1W1} [W^{\alpha,1}_n,W^{\beta,1}_m]&=f^{\alpha\beta}_\gamma
\Bigl(2\cos\pi\ffrac{3n-3m}{2N}W^{\gamma,4}_{n+m}-2\cos\pi\ffrac{3n+3m}{2N}W^{\gamma,4}_{n-m}-\\
&\notag\kern30pt -\half \bigl(\cos\pi\ffrac{n-m}{2N}+3\cos\pi\ffrac{3n-3m}{2N}\bigr)W^{\gamma,2}_{n+m}+\\
&\notag\kern70pt +\half\bigl(\cos\pi\ffrac{n+m}{2N}+3\cos\pi\ffrac{3n+2m}{2N}\bigr)W^{\gamma,2}_{n-m}+\\
&\notag\kern110pt +\half \cos\pi\ffrac{n-m}{2N}\sin\pi\ffrac{n}{N}\sin\pi\ffrac{m}{N}W^{\gamma,0}_{n+m}+\\
&\notag\kern150pt +\half\cos\pi\ffrac{n+m}{2N}\sin\pi\ffrac{n}{N}\sin\pi\ffrac{m}{N}W^{\gamma,0}_{n-m}\Bigr)\\
&\notag+g^{\alpha\beta}\Bigl(4\sin\pi\ffrac{3n-3m}{2N}H^5_{n+m}-4\sin\pi\ffrac{3n+3m}{2N}H^5_{n-m}-\\
&\notag\kern20pt-\bigl(\sin\pi\ffrac{n-m}{2N}+5\sin\pi\ffrac{3n-3m}{2N}\bigr)H^3_{n+m}+\\
&\notag\kern45pt+\bigl(\sin\pi\ffrac{n+m}{2N}+5\sin\pi\ffrac{3n+3m}{2N}\bigr)H^3_{n-m}-\\
&\notag\kern65pt-\ffrac{1}{4}\bigl(\cos\pi\ffrac{n+m}{N}\sin\pi\ffrac{n-m}{2N}
-3\sin\pi\ffrac{n-m}{2N}-5\sin\pi\ffrac{3n-3m}{2N}\bigr)H^1_{n+m}+\\
&\notag\kern85pt+\ffrac{1}{4}\bigl(\cos\pi\ffrac{n-m}{N}\sin\pi\ffrac{n+m}{2N}
-3\sin\pi\ffrac{n+m}{2N}-5\sin\pi\ffrac{3n+3m}{2N}\bigr)H^1_{n-m}\Bigr).
\end{align}
We note that the commutator $[W^{\alpha,1}_n,W^{\beta,0}_m]$ can be obtained from $[W^{\alpha,0}_n,W^{\beta,1}_m]$ by replacing $\alpha$ with $\beta$
and $n$ with $m$ and noting that $H^r_{-n}=(-1)^rH^r_n$ and $W^r_{-n}=(-1)^rW^r_n$.

\end{document}